%% file: main.tex
\documentclass[runningheads,usenames,dvipsnames,envcountsame]{llncs}
\usepackage{graphicx}
\usepackage{amsfonts}
\usepackage{amsmath}
\usepackage{graphicx}
\usepackage[colorlinks=true, allcolors=blue]{hyperref}
\usepackage{lipsum}
\usepackage{lmodern}
\usepackage{tcolorbox}
\usepackage{algorithm}
\usepackage[noend]{algpseudocode}
\usepackage{xspace}
\usepackage{dsfont}
\usepackage{soul}
\usepackage{tabularray}
\usepackage{subcaption}
\usepackage{cite}
\usepackage{soul}
\usepackage{booktabs}

\newcommand{\qedhere}{\qed}

\input{Preamble.tex}

\begin{document}
\title{The Power of Amortization\\on \pw{Minimizing Total Completion Time} with Explorable Uncertainty\thanks{\pw{Preliminary versions of this paper appeared in WAOA 2023~\cite{DBLP:conf/waoa/LiuLWZ23} and MAPSP 2024~\cite{MAPSP2024}.}}}
\titlerunning{\pw{The Power of Amortization on Explorable Uncertainty}}
%
\author{Bob Krekelberg\inst{1}\orcidID{0009-0000-5517-6095} \and
Alison Hsiang-Hsuan Liu\inst{1}\orcidID{0000-0002-0194-9360} \and
Fu-Hong Liu\orcidID{0000-0001-6073-8179} \and
Prudence W.H. Wong\inst{2}\orcidID{0000-0001-7935-7245}\thanks{The work is partially supported by University of Liverpool Covid Recovery Fund and Royal Society International Exchanges.} \and
Xiao-Ou Zhang\inst{1}}
%
\authorrunning{B. Krekelberg et al.}
%
\institute{Utrecht University, the Netherlands \and
University of Liverpool, United Kingdom}
%
\maketitle              

\begin{abstract}
\input{Abstract}
\end{abstract}

\section{Introduction}
\input{Intro}

\section{\pw{Preliminaries}}\label{sec:prelim}
\input{Prelim}

\section{Deterministic algorithms on a single machine}\label{sec:det}
\input{Amortization}

\input{NewAlgo}
\input{Prmp}

\section{Randomized algorithm on a single machine}\label{sec:rand}
\input{Randomized}

\section{\pw{Extension} to multiple machines}\label{sec:parallel}
\input{Parallel}

\section{Conclusion}\label{sec:conclusion}
\input{Conclusion}

\bibliographystyle{splncs04}
\bibliography{reference}

\appendix
\section{Appendix}
\input{Appendix}

\end{document}

%% file: Preamble.tex
\newcommand{\fullversion}[1]{{#1}}
\newcommand{\shortversion}[1]{}

\newcommand{\bk}[1]{{\color{Purple}{#1}}}
\newcommand{\todo}[1]{{\color{Red}{TODO: #1}}}

\newcommand{\pw}[1]{{#1}}
\newcommand{\hhl}[1]{{#1}}
\newcommand{\fhl}[1]{{#1}}
\newcommand{\hhldel}[1]{}
\newcommand{\tableemph}[1]{{\color{red}{$\mathbf{#1}$}}}

\newcommand{\runtitle}[1]{\textbf{#1}}
\newcommand{\hide}[1]{}
\newcommand{\hideproof}[1]{}

\newcommand{\cost}{cost}
\newcommand{\optimal}{\text{OPT}}
\newcommand{\AEalgo}{(\alpha,\beta)\text{-SORT}} 
\newcommand{\ouralgo}{(\alpha, \beta)\text{-PCP}}
\newcommand{\randalgo}{\beta\text{-RPCP}}
\newcommand{\paralgo}{(\alpha,\beta)\text{-PCPM}}
\newcommand{\randparalgo}{\beta\text{-RPCPM}}

\newcommand{\opt}[1]{{#1^*}}
\newcommand{\alg}[1]{{#1^A}}
\newcommand{\pr}{\textsc{SUP}}
\newcommand{\prob}{\mathds{P}}
\newcommand{\expect}{\mathds{E}}

\newcommand{\factorOne}{R}
\newcommand{\factorTwo}{r}

\newcommand{\commenttex}[1]{}

%% file: Abstract.tex

We study online scheduling to minimize total completion time with explorable uncertainty on single and multiple machines.
Each job comes with an upper limit of its processing time, 
which could be potentially reduced by testing the job, which also takes time.
The objective is to schedule all jobs with minimum total completion time.
The challenge lies in deciding which jobs to test, the order of testing/processing jobs, and in multiple machine case which machine a job is allocated to.
In multiple machine case, testing and processing of a job are allowed to be scheduled on different machines.

\hide{
The online problem was first introduced on a single machine with unit testing time~\cite{DBLP:conf/innovations/DurrEMM18,DBLP:journals/algorithmica/DurrEMM20}
and later generalized to variable testing times~\cite{DBLP:conf/waoa/AlbersE20}.
For this general setting, the upper bounds of the competitive ratio are shown to be
$4$ and $3.3794$ for deterministic and randomized online algorithms~\cite{DBLP:conf/waoa/AlbersE20}, respectively;
while the lower bounds for unit testing time stands~\cite{DBLP:conf/innovations/DurrEMM18,DBLP:journals/algorithmica/DurrEMM20}, 
which are $1.8546$ (deterministic) and $1.6257$ (randomized).
Recent work has extended to multiple machines~\cite{DBLP:journals/algorithmica/GongCH23}.
}

Different settings have been studied before.
In this work, we first consider the variable testing times setting.
We enhance the analysis framework in Albers and Eckl (2020)
and improve the analysis of the competitive ratio of their deterministic single machine algorithm 
from $4$ to $1+\sqrt{2} \approx 2.4143$.
Using the new analysis framework, we propose a new deterministic algorithm that further improves the competitive ratio to $2.316513$.
The new framework also enables us to develop a randomized algorithm improving the expected competitive ratio from $3.3794$ to $2.152271$.
We further show that with $m$~machines, by extending the framework of Gong et al. (2024),
there exists a deterministic $2.77629-(0.45977/m)$-competitive algorithm and a randomized $2.51098-(0.3587/m)$-competitive algorithm. 
The performance of the algorithms on multiple machines when $m = 1$ matches the current best algorithms on a single machine for variable testing times shown in this paper.

%% file: Intro.tex



\subsection{Background and motivation}\label{sec:background}

In this work, we study the \emph{Scheduling with Uncertain Processing time} (\pr) problem with the minimized total completion time objective.
We are given $n$ jobs, where each job has a \emph{testing time} $t_j$ and an \emph{upper limit} $u_j$ of its \emph{real processing time} $p_j \in [0,u_j]$.
A job $j$ can be executed {(without testing)}, taking $u_j$ time units. 
A job $j$ can also be tested using $t_j$ time units, and after it is tested, it takes $p_j$ time to execute.
Note that any algorithm needs to test a job $j$ beforehand to run it in time $p_j$.
The online algorithm does not know the exact value of $p_j$ unless it tests the job. 
On the other hand, the optimal offline algorithm knows in advance each $p_j$ even before testing. 
Therefore, the optimal strategy is to test job $j$ if and only if $t_j+p_j \leq u_j$ and execute the shortest job first, where the processing time of a job $j$ is $\min\{t_j+p_j, u_j\}$~\cite{DBLP:conf/waoa/AlbersE20,DBLP:conf/innovations/DurrEMM18,DBLP:journals/algorithmica/DurrEMM20}.
However, since the online algorithm only learns about $p_j$ after testing $j$, the challenge to the online algorithm is to decide which jobs to test and the order of tasks that could be testing, execution, or execution-untested.

\medskip

It is typical to study uncertainty in scheduling problems, for example,
in the worst case scenario for online or stochastic optimization.
Kahan~\cite{DBLP:conf/stoc/Kahan91} has introduced a novel notion of explorable uncertainty
where queries can be used to obtain additional information with a cost.
The model of scheduling with explorable uncertainty studied in this paper was introduced by D{\"u}rr et al. recently~\cite{DBLP:conf/innovations/DurrEMM18,DBLP:journals/algorithmica/DurrEMM20}.
In this model, job processing times are uncertain in the sense that only an upper limit of the processing time is known, and can be reduced potentially by testing the job, which takes a testing time that may vary according to the job.
An online algorithm does not know the real processing time before testing the job, whereas an optimal offline algorithm has the full knowledge of the uncertain data.

One of the motivations to study scheduling with uncertain processing time is clinic scheduling~\cite{DBLP:conf/aiia/CarusoGMMP21,DBLP:conf/data/LopesVVFSS23}. Without a pre-diagnosis, it is safer to assign each treatment the maximum time it may need.
With pre-diagnosis, the precise time a patient needs can be identified, which can improve the performance of the scheduling. 
Other applications are, 
as mentioned in~\cite{DBLP:conf/innovations/DurrEMM18,DBLP:journals/algorithmica/DurrEMM20},
code optimization~\cite{book/CardosoDC17}, compression for file transmission over network~\cite{DBLP:journals/sigops/WisemanSW05},
fault diagnosis in maintenance environments~\cite{inbook/NicolaiD08}.
Application in distributed databases with centralized master server~\cite{DBLP:conf/vldb/OlstonW00} is also discussed in~\cite{DBLP:conf/waoa/AlbersE20}.

\medskip

In addition to its practical motivations, the model of explorable uncertainty also blurs the line between offline and online problems by allowing a restricted uncertain input. It enables us to investigate how uncertainty influences online decision quality in a more quantitative way.
The concept of exploring uncertainty has raised a lot of attention and has been studied on different problems, such as sorting~\cite{DBLP:journals/tcs/HalldorssonL21}, finding the median~\cite{DBLP:conf/stoc/FederMPOW00}, identifying a set with the minimum-weight among a given collection of feasible sets~\cite{DBLP:journals/tcs/Erlebach0K16}, finding shortest paths~\cite{DBLP:journals/jal/FederMOOP07}, computing minimum spanning trees~\cite{DBLP:conf/stacs/HoffmannEKMR08}, etc. 
More recent work and a survey can be found in~\cite{DBLP:journals/jal/FederMOOP07,DBLP:journals/mst/GuptaSS16,DBLP:journals/eatcs/Erlebach015}.
Note that in many of the works, the aim of the algorithm is to find the optimal solution with the minimum number of testings for the uncertain input, comparing against the optimal number of testings. 

Another closely related model is Pandora's box problem~\cite{dc603527-709e-3110-a352-6b21b157c70c,DBLP:conf/aaai/EsfandiariHLM19,DBLP:conf/soda/DingFHTX23}, which was based on the secretary problem, that was first proposed by Weitzman~\cite{dc603527-709e-3110-a352-6b21b157c70c}. In this problem, each candidate (that is, the box) has an independent probability distribution for the reward value. To know the exact reward a candidate can provide, one can open the box and learn its realized reward. 
More specifically, at any time, an algorithm can either open a box, or select a candidate and terminate the game. 
However, opening a box costs a price. The goal of the algorithm is to maximize the reward from the selected candidate minus the total cost of opening boxes.
The Pandora's box problem is a foundational framework for studying how the cost of revealing uncertainty affects the decision quality. 
More importantly, it suggests what information to acquire next after gaining some pieces of information.


\subsection{Previous works}\label{sec:previous_work}

\runtitle{Single machine.}
For the $\pr$ problem, D{\"u}rr et al. studied the case where all jobs have the same testing time~\cite{DBLP:conf/innovations/DurrEMM18,DBLP:journals/algorithmica/DurrEMM20}. 
In the paper, the authors proposed a \textsc{Threshold} algorithm for the special instances. 
For the competitive analysis, the authors proposed a delicate \emph{instance-reduction} framework. 
Using this framework, the authors showed that the worst case instance of \textsc{Threshold} has a special format. An upper bound of the competitive ratio of $2$ of \textsc{Threshold} is obtained by the ratio of the special format instance.
Using the instance-reduction framework, the authors also showed that when all jobs have the same testing time and the same upper limit, there exists a $1.9338$-competitive \textsc{Beat} algorithm.
The authors provided a lower bound of $1.8546$ for any deterministic online algorithm. 
For randomized algorithms, the authors showed that the expected competitive ratio is between $1.6257$ and $1.7453$.


Later, Albers and Eckl~\cite{DBLP:conf/waoa/AlbersE20} studied a more general case where jobs have variable testing time. 
The authors proposed a classic and elegant framework where the completion time of an algorithm is divided into contribution segments by the jobs executed prior to it. 
For the jobs with ``correct'' execution order as they are in the optimal solution, their total contribution to the total completion time is charged to twice the optimal cost by the fact that the algorithm does not pay too much for wrong decisions of testing a job or not.
For the jobs with ``wrong'' execution order, their total contribution to the total completion time is charged to another twice the optimal cost using a \emph{comparison tree} method, which is bound with the proposed $\AEalgo$ algorithm.
The authors showed that for $\alpha=1$ and $\beta=1$, the algorithm is at most $4$-competitive.
The authors also provided a preemptive $3.2361$-competitive algorithm and an expected $3.3794$-competitive randomized algorithm.
Furthermore, the authors showed that the deterministic lower bound of 1.8546 in~\cite{DBLP:conf/innovations/DurrEMM18,DBLP:journals/algorithmica/DurrEMM20} also holds in the preemptive case.

\runtitle{Multiple machines.}
Gong et al.~\cite{DBLP:journals/algorithmica/GongCH23} extended the framework introduced by Albers and Eckl~\cite{DBLP:conf/waoa/AlbersE20} to the setting of $m$ multiple machines.
Their approach is applicable to both the multiple and single-machine cases.
For the case in which testing times are unit,
they proposed an algorithm whose competitive ratio for $m\geq 5$ is bounded  by  $\phi + \frac{\phi + 1}{2} \cdot (1 - \frac{1}{m})$, where $\phi = \frac{1+\sqrt{5}}{2}$ is the golden ratio.
This ratio converges to 2.9271 as $m$ approaches infinity.
For $m < 5$, the algorithm has a competitive ratio of $\phi+1<2.6181$.
When testing times are variable, they introduced an algorithm with a competitive ratio of $2\phi < 3.2361$.
Moreover, they developed a preemptive algorithm that achieves a competitive ratio of $3$ on one or two machines, and a ratio of $3.5 - \frac{3}{2m}$ for $m\geq 3$.
It is noteworthy that this latter result is outperformed by the non-preemptive $2\phi$-competitive algorithm when $m\geq 6$.
For randomized algorithms, Gong et al. proposed an algorithm with a competitive ratio of $2.8307 - \frac{1.3962}{m}$ when $m\geq 37$, and a ratio of 2.7925 otherwise.
A summary of previous results for single and multiple machines can be found in Table~\ref{tb:results}.




\runtitle{Other settings.}
The $\pr$ problem has also been studied with other objective, namely minimizing maximum completion time, in single machine~\cite{DBLP:conf/waoa/AlbersE20,DBLP:conf/innovations/DurrEMM18,DBLP:journals/algorithmica/DurrEMM20}, multiple machines~\cite{DBLP:conf/wads/AlbersE21,DBLP:conf/faw/GongL21,DBLP:journals/jco/GongGLM22}. Recently, obligatory testing has also been studied~\cite{DBLP:conf/esa/DogeasEL24}.

\subsection{Our contributions}\label{sec:contributions}
Our main contribution is to enhance the analysis framework of algorithms for the $\pr$ problem using amortization.
The enhanced framework, applicable in both single and multiple machine settings, allows us to improve the analysis of existing algorithm and develop new algorithms with better performance.

We first analyze the $\AEalgo$ algorithm proposed in the work~\cite{DBLP:conf/waoa/AlbersE20} in a more amortized sense.
Instead of charging the jobs in the correct order and in the wrong order to the optimal cost separately, we manage to partition the tasks into groups and charge the total cost in each of the groups to the optimal cost regarding the group.
The introduction of amortization to the analysis creates room for improving the competitive ratio by adjusting the values of $\alpha$ and $\beta$. 
The possibility of picking $\alpha > 1$ helps balance the penalty incurred by making a wrong guess on testing a job or not.
On the other hand, the room for different $\beta$ values allows one to differently prioritize the tasks that provide extra information and the tasks that immediately decide a completion time for a job. 
By this new analysis and the room of choosing different values of $\alpha$ and $\beta$, we improve the upper bound of the competitive ratio of $\AEalgo$ from $4$ to $1+\sqrt{2}$. 

With the power of amortization, we improve the algorithm by further prioritizing different tasks using different parameters. 
We show that our new algorithm, called $\ouralgo$, is $2.316513$-competitive. 
This algorithm is extended to a randomized version, called $\randalgo$ (prefix R for randomized), with an expected competitive ratio of $2.152271$. 
We further show that under the current problem setting, preempting the execution of jobs does not help in gaining a better algorithm.

We further extend our analysis framework to multiple machines and combine with the technique developed in~\cite{DBLP:journals/algorithmica/GongCH23} to obtain a deterministic algorithm, called $\paralgo$ (suffix M for multiple machines), with competitive ratio $2.77629-(0.45977/m)$;
as well as a randomized algorithm, called $\randparalgo$, with expected competitive ratio $2.51098-(0.3587/m)$.
These results on variable testing times when $m=1$ matches the current best algorithms on a single machine.
A summary of our results and comparison with existing results can be found in Table~\ref{tb:results}.

\begin{table}[htbp]
\centering
\begin{tabular}{lllll}
\toprule
              & Testing      & Upper       & Upper       & Lower       \\
              & time         & limit       & bound       & bound       \\
\midrule
Det.
              & $1$         & Uniform     & $1.9338$~\cite{DBLP:conf/innovations/DurrEMM18,DBLP:journals/algorithmica/DurrEMM20} & $1.8546$~\cite{DBLP:conf/innovations/DurrEMM18,DBLP:journals/algorithmica/DurrEMM20} \\
\cmidrule(lr){3-4}
($m=1$)       &             & Variable    & $2$~\cite{DBLP:conf/innovations/DurrEMM18,DBLP:journals/algorithmica/DurrEMM20} & \\
\cmidrule(lr){2-4}
              & Variable    & Variable    & $4$~\cite{DBLP:conf/waoa/AlbersE20} \tableemph{\rightarrow 2.414} (Thm.~\ref{Thm:amortized}, $\AEalgo$) & \\
              &             &             & $3.2361$~\cite{DBLP:journals/algorithmica/GongCH23} & \\
              &             &             & \tableemph{2.316513} (Thm.~\ref{Thm:improved}, $\ouralgo$) & \\
\cmidrule(lr){4-4}
              &             & (Prmp.)     & $3$~\cite{DBLP:journals/algorithmica/GongCH23} & \cite{DBLP:conf/waoa/AlbersE20}\footnotemark \\
              &             &             & \tableemph{2.316513} (Thm.~\ref{Thm:improved}, $\ouralgo$) & \\
\cmidrule(lr){2-4}
($m \geq 2$)   & $1$         & Variable    & $(m\geq 5)$ $2.9271 - (1.3090/m)$~\cite{DBLP:journals/algorithmica/GongCH23} & \\
              &             &             & \tableemph{2.73606 - (0.5/m)} (Thm.~\ref{Thm:multi-p_uniform}) & \\
\cmidrule(lr){2-4}
              & Variable    & Variable    & $3.2361$~\cite{DBLP:journals/algorithmica/GongCH23} & \\
              &             &             & \tableemph{2.77629 - (0.45977/m)} (Thm.~\ref{Thm:multi-p_variable}, $\paralgo$) & \\
\cmidrule(lr){4-4}
              &             & (Prmp.)     & $(m\geq 6)$ 3.2361~\cite{DBLP:journals/algorithmica/GongCH23} & \\
\midrule
Rand.
              & $1$         & Variable    & $1.7453$~\cite{DBLP:conf/innovations/DurrEMM18,DBLP:journals/algorithmica/DurrEMM20} & $1.6257$~\cite{DBLP:conf/innovations/DurrEMM18,DBLP:journals/algorithmica/DurrEMM20} \\
\cmidrule(lr){2-4}
($m=1$)       & Variable    & Variable    & $2.7925$~\cite{DBLP:journals/algorithmica/GongCH23} & \\
              &             &             & \tableemph{2.152271} (Thm.~\ref{Thm:randomized}, $\randalgo$) & \\
\cmidrule(lr){2-4}
($m \geq 2$)   & Variable    & Variable    & ($m\geq 37$) $2.8307 - (1.3962/m)$~\cite{DBLP:journals/algorithmica/GongCH23} & \\
              &             &             & \tableemph{2.51098 - (0.3587/m)} (Thm.~\ref{Thm:randomized_parallel}, $\randparalgo$) & \\
\toprule
\end{tabular}
\caption{Summary of deterministic (Det.) and randomized (Rand.)  results with different settings on testing time and upper limit of processing time. The results from this work are in bold and red.}
\label{tb:results}
\end{table}
\footnotetext{Albers and Eckl proved that the lower bound of 1.8546 also holds in the preemptive case.}
\medskip

\runtitle{Organization of the paper.}
In Section~\ref{sec:prelim}, we introduce the notations used in this paper. We also review the algorithm and analysis of the $\AEalgo$ algorithm proposed in the work~\cite{DBLP:conf/waoa/AlbersE20}.
In Section~\ref{sec:det}, we elaborate on how amortized analysis helps to improve the competitive analysis of $\AEalgo$ (Section~\ref{subsec:amortization}).
Upon the new framework, we propose a better algorithm, $\ouralgo$ in Section~\ref{subsec:improved} and
show the tightness of our analysis in Section~\ref{sec:algorithm_lower_bound}.
In Section~\ref{subsec:preemption}, we argue that the power of preemption is limited in the current model.
In Section~\ref{sec:rand}, we show how amortization helps to improve the performance of randomized algorithms.
The extension to multiple machines is given in Section~\ref{sec:parallel}.
Finally, we conclude in Section~\ref{sec:conclusion}.
\shortversion{
For the sake of the page limit, we leave the proofs in the full version.
}

%% file: Prelim.tex

Given $n$ jobs $1, 2, \cdots, n$, each job $j$ has a \emph{testing time} $t_j$ and an \emph{upper limit} $u_j$ of its \emph{real processing time} $p_j \in [0,u_j]$.
A job $j$ can be executed-untested in $u_j$ time units or be tested using $t_j$ time units and then executed in $p_j$ time units.
The \emph{tasks} regarding a job $j$ are the testing, execution, or execution-untested of~$j$ (taking $t_j$, $p_j$, or $u_j$ time, respectively).
Note that if a job is tested, it does not need to be executed immediately.
{That is, for a tested job, there can be tasks regarding other jobs between its testing and its execution.}
Furthermore, its execution task can be assigned to different machines from the testing task on multiple machine case.

We denote by~$\alg{p_j}$ the time spent by an algorithm~$A$ on job~$j$, i.e., $\alg{p_j} = t_j+p_j$ if $A$ tests $j$, and $\alg{p_j} = u_j$ otherwise.
Let $\optimal$ denote the optimal algorithm.
Similarly, we denote by~$\opt{p_j}$ the time spent by~$\optimal$.
Since $\optimal$ knows $p_j$ in advance, it can decide optimally whether to test a job, i.e., $\opt{p_j}=\min\{u_j, t_j+p_j\}$, and execute the jobs in non-decreasing order of $\opt{p_j}$ when a machine is available.
We denote by $cost(A)$ the total completion time of any algorithm~$A$.

We follow the notation in the work of Albers and Eckl~\cite{DBLP:conf/waoa/AlbersE20} and denote by $c(k,j)$ the \emph{contribution} of job $k$ in the completion time of job $j$ in the online schedule $A$. 
That is, $c(k,j)$ is the total time of the tasks regarding job $k$ before the completion time of job $j$. 
The \emph{completion time} of job $j$ in the schedule $A$, denoted by $c_j$, is then $\sum_{k=1}^n c(k,j)$.
Similarly, we define $\opt{c}(k,j)$ as the contribution of job $k$ in the completion time of job $j$ in the optimal schedule.
\hide{\hhl{Since the optimal algorithm knows $p_j$ in advance, it can decide optimally whether to test a job and execute the jobs in the shortest $\opt{p_j}$ processing time first order.}}
{As observed, $\optimal$ schedules in non-decreasing order of $p^*$,} $\opt{c}(k,j) = 0$ if $k$ is executed after $j$ in the optimal schedule, and $\opt{c}(k,j) = \opt{p_k}$ otherwise. 


We denote by $i <_o j$ if the optimal schedule executes job $i$ before job $j$. 
We also define $i >_o j$ and $i =_o j$ similarly (in the latter case, job $i$ and job $j$ are the same job).
On a single machine, the completion time of job $j$ in the optimal schedule is denoted by $\opt{c_{j}} = \sum_{i\leq_o j} \opt{p_{i}}$.
The total completion time of the optimal schedule is then $\sum_{j=1}^n \opt{c_j}$. Note that there is an optimal strategy where $\opt{p_i} \leq \opt{p_j}$ if $i \leq_o j$.


\subsection{Review of $\AEalgo$ algorithm~\cite{DBLP:conf/waoa/AlbersE20} on a single machine}
\label{subsec:review}
For completeness, we summarize the $\AEalgo$ algorithm and its analysis proposed in {the work of Albers and Eckl}~\cite{DBLP:conf/waoa/AlbersE20}.

\smallskip

Intuitively, the algorithm tests a job $j$ if and only if $u_j\geq \alpha\cdot t_j$.
Depending on whether a job is tested or not, the job is transformed into one task ({execution-untested} task) or two tasks (testing task and execution task) with weights.
These tasks are then maintained in a priority queue for the algorithm to decide their processing order (tasks with smaller weight have higher priority).
More specifically, a testing task has a weight of $\beta\cdot t_j$, an execution task has a weight of $p_j$, and an {execution-untested} task has a weight of $u_j$. (See Algorithm~\ref{Alg:Susanne}.)
After all initial tasks (namely, all testing tasks for tested jobs and all execution-untested tasks for untested jobs are inserted into the priority queue, the algorithm executes the tasks in the queue and deletes the executed tasks, starting from the task with the shortest (weighted) time. 
If the task is a testing of a job $j$, the resulting $p_j$ is inserted as an execution task into the queue after testing. (See Algorithm~\ref{Alg:queue}.)
Intuitively, both $\alpha$ and $\beta$ are at least $1$. The precise values of $\alpha$ and $\beta$ will be decided later based on the analysis.

\begin{algorithm}[t]
    \caption{$\AEalgo$ algorithm~\cite{DBLP:conf/waoa/AlbersE20}}
    \label{Alg:Susanne}
    \begin{algorithmic}
        \State Initialize a priority queue $Q$ where tasks with smaller weight have higher priority
        \For{$j = 1, 2, 3, \cdots, n$}
            \If{$u_j\geq \alpha\cdot t_j$}
                \State Insert a testing task with weight $\beta\cdot t_j$ into $Q$
            \Else
                \State Insert an execution-untested task with weight $u_j$ into $Q$
            \EndIf
        \EndFor
    \State \textbf{Queue-Execution}$(Q)$ \Comment{See Algorithm~\ref{Alg:queue}}
    \end{algorithmic}
\end{algorithm}

\begin{algorithm}[t]
\begin{algorithmic}
    \Procedure{Queue-Execution}{$Q$}
        \While{$Q$ is not empty}
            \State $x \gets$ Extract the smallest-weight task in $Q$
            \If{$x$ is a testing task for a job $j$}
                \State {Test job $j$}
                \Comment{It takes $t_j$ time}
                \State {Insert an execution task with weight $p_j$ into $Q$}
            \ElsIf{$x$ is an execution task for a job $j$}
                \State Execute (tested) job $j$
                \Comment{It takes $p_j$ time}
            \Else \Comment{$x$ is an execution-untested task for a job $j$}
                \State Execute job $j$ untested
                \Comment{It takes $u_j$ time}
            \EndIf
        \EndWhile
    \EndProcedure
\caption{Procedure \textbf{Queue-Execution} $(Q)$}
\label{Alg:queue}
\end{algorithmic}
\end{algorithm}

\runtitle{Analysis~\cite{DBLP:conf/waoa/AlbersE20}.}
Recall that $c(k,j)$ is the contribution of job $k$ of the completion time of job $j$, and the completion time of job $j$ is $\alg{c_j} = \sum_{k = 1}^n c(k, j)$. 
The key idea of the analysis is that given job $j$, partitioning the jobs (say, $k$) that are executed before $j$ into two groups, $k\leq_o j$ or $k>_oj$.
Since the algorithm only tests a job $j$ when $u_j\geq \alpha t_j$, $\alg{p_k} \leq \max\{\alpha, 1+\frac{1}{\alpha}\}\cdot\opt{p_k}$. 
Therefore, the total cost incurred by the first group of jobs is at most $\max\{\alpha, 1+\frac{1}{\alpha}\}\cdot\cost(\optimal)$. 
Note that the ratio, in this case, reflects the penalty to the algorithm that makes a wrong guess on testing a job or not.

\smallskip

For the second group of jobs, the authors proposed a classic and elegant \emph{comparison tree} framework to charge each $c(k,j)$ with $k>_o j$ to the time that the optimal schedule spends on job $j$. 
More specifically, $c(k,j) \leq \max\{(1+\frac{1}{\beta})\alpha, 1+\frac{1}{\alpha}, 1+\beta\}\cdot \opt{p_j}$ for any $k$ and $j$. 
Hence, the total cost incurred by the second group of jobs can be charged to $\max\{(1+\frac{1}{\beta})\alpha, 1+\frac{1}{\alpha}, 1+\beta\}\cdot\cost(\optimal)$.


\smallskip

By summing up the $c(k, j)$ values for all pairs of $k$ and $j$, the total completion time of the algorithm is at most 
{\[\max\{\alpha, 1+\frac{1}{\alpha}\} + \max\{(1+\frac{1}{\beta})\cdot\alpha, 1 + \frac{1}{\alpha}, 1+\beta\}.\]}
When $\alpha = \beta = 1$ (which is the optimal selection), the competitive ratio is $4$.

\subsection{Our observations}
\label{subsec:observation}
As stated by Albers and Eckl~\cite{DBLP:conf/waoa/AlbersE20}, $\alpha=\beta=1$ is the optimal choice in their analysis framework.
Therefore, it is not possible to find a better $\alpha$ and $\beta$ to tighten the competitive ratio under their analysis framework.
However, the framework can be improved via the following observations.

\smallskip

For example, given that $\alpha=\beta=1$, consider two jobs $k$ and $j$, where $(t_k, u_k, p_k) = (1+\varepsilon, 1+3\varepsilon, 1+3\varepsilon)$ and $(t_j, u_j, p_j) = (1, 1+4\varepsilon, 1+2\varepsilon)$.
By the $\AEalgo$ algorithm, both $k$ and $j$ are tested.  The order of the tasks regarding these two jobs is $t_j$, $t_k$, $p_j$, and finally $p_k$.
On the other hand, in the optimal schedule, $\opt{p_k} = u_k = 1+3\varepsilon$ and $\opt{p_j} = u_j = 1+4\varepsilon$. 
Since $k\leq_o j$, as shown in Figure~\ref{fig:observation}, both $c(k,j)$ and $c(j,k)$ are charged to $2\opt{p_k}$, separately. 
Note that although $c(k,j) = t_k$ in this example, the worst-case nature of the analysis framework fails to capture the fact that the contribution from the tasks regarding $k$ to the completion time of $j$ is even smaller than $\opt{p_k}$.
{This observation motivates us to establish a new analysis framework.}

\hide{For example, given that $\alpha=\beta=1$, consider two jobs $k$ and $j$, where $(t_k, u_k, p_k) = (1.1, 1.3, 1.3)$ and $(t_j, u_j, p_j) = (1, 1.4, 1.2)$. 
By the $\AEalgo$ algorithm, both $k$ and $j$ are tested.  The order of the tasks regarding these two jobs is $t_j$, $t_k$, $p_j$, and finally $p_k$.
On the other hand, in the optimal schedule, $\opt{p_k} = u_k = 1.3$ and $\opt{p_j} = u_j = 1.4$. 
Since $k\leq_o j$, as shown in Figure~\ref{fig:observation}, both $c(k,j)$ and $c(j,k)$ are charged to $2\opt{p_k}$, separately. 
Note that although $c(k,j) = t_k$ in this example, the worst-case nature of the analysis framework fails to capture the fact that the contribution from the tasks regarding $k$ to the completion time of $j$ is even smaller than $\opt{p_k}$.
{This observation motivates us to establish a new analysis framework.}}

\begin{figure}
    \centering
    \includegraphics[width=5cm]{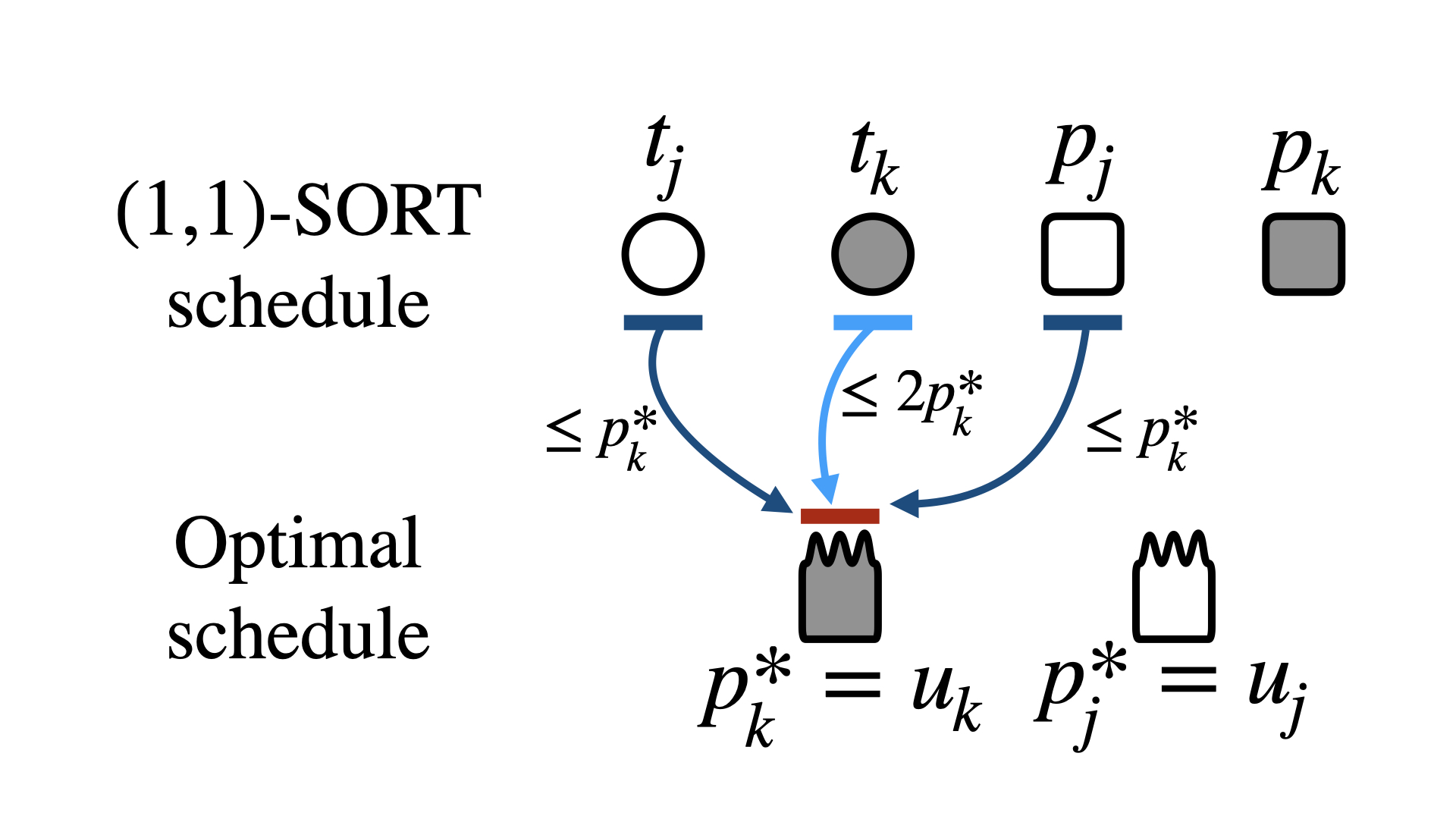}
    \caption{An example where $\opt{p_k}$ is charged four times. The light blue and dark blue segments represent $c(k,j)$ and $c(j,k)$, respectively. The red segment represents $\opt{p_k}$.}
    \label{fig:observation}
\end{figure}

\hide{
\fhl{\st{Given $\alpha=\beta=1$}}, consider jobs $k$ and $j$ where $k\leq_o j$. When $\alg{c_j}$ is estimated, $c(k, j)$ is charged to $2\cdot\opt{p_k}$. Meanwhile, when $\alg{c_k}$ is estimated, $c(j, k)$ is charged to $2\cdot \opt{p_k}$ again.
More specifically, there are four cases where $c(k,j)$ is bounded by $2\cdot \opt{p_k}$:
\begin{enumerate}
    \item When $k$ is tested\fhl{,} $j$ is not tested, $\beta\cdot t_k\leq u_j$, and $p_k \leq u_j$, 
    \item When both $k$ and $j$ are tested, $t_k \leq t_j$, and $p_k \leq \beta\cdot t_j$,
    \item When both $k$ and $j$ are tested, $t_k \leq t_j$, and $\beta\cdot t_j \leq p_k \leq p_j$, or
    \item When both $k$ and $j$ are tested, $t_j \leq t_k$, and $p_k \leq p_j$.
\end{enumerate}

It can be observed that no two jobs can be $k$ and $j$ at the same time. \todo{(Figure?)}
}

%% file: Amortization.tex
\hide{
As mentioned by Albers and Eckl~\cite{DBLP:conf/waoa/AlbersE20}, the choice where $\alpha=\beta=1$ is optimal in the analysis framework. 
However, the observation in {Section}~\ref{subsec:observation} suggests that there is a slack between the estimated
{and the actual total completion time.}

\smallskip
}
In this section, we consider scheduling on a single machine.
We first enhance the framework by equipping it with amortized analysis in Section~\ref{subsec:amortization}. 
Using amortized arguments, for any two jobs $k\leq_o j$, we will show how to charge the sum of $c(k,j) + c(j,k)$ to $\opt{p_k}$.  
The new framework not only improves the competitive ratio but also creates room for adjusting $\alpha$ and $\beta$.

\smallskip

Then in Section~\ref{subsec:improved}, we propose an algorithm called $\ouralgo$ to improve the $\AEalgo$ algorithm based on our enhanced framework.

\subsection{Amortization}
\label{subsec:amortization}
\hide{
Our first attempt is to introduce amortization into the analysis of the $\AEalgo$ algorithm.
The core idea is, instead of charging $c(k,j)$ (the contribution of job $k$ to the completion time of job $j$) for each pair of jobs $k$ and $j$ separately, we charge $c(k,j) + c(j,k)$ to $\opt{c}(k,j) + \opt{c}(j,k)$.

\smallskip

\runtitle{Analysis framework.}}
We first bound $c(k,j)+c(j,k)$ for all pairs of jobs $k$ and $j$ with $k\leq_o j$ by a function $r(\alpha,\beta) \cdot \opt{c}(k,j)$. Then, we can conclude that the algorithm is $r(\alpha,\beta)$-competitive by the following argument:
\begin{align*}
    \cost(\AEalgo) &=\sum_{j=1}^n \sum_{k=1}^n c(k,j) 
    = \sum_{j=1}^n (\sum_{k<_o j} (c(k,j) + c(j,k)) + c(j,j))\\
    & \leq \sum_{j=1}^n r(\alpha,\beta) \cdot (\sum_{k<_o j}\opt{c}(k,j)+ \opt{c}(j,j))\\
    &= r(\alpha,\beta) \cdot \cost(\optimal)
\end{align*}

\smallskip

To bound $c(k,j)+c(j,k)$ by the cost of tasks 
$k$, we first observe that it is impossible that $c(k,j) = \alg{p_k}$ and $c(j,k) = \alg{p_j}$ at the same time. 
More specifically, depending on whether the jobs $k$ and $j$ are tested or not, the last task regarding these two jobs does not contribute to $c(k,j) + c(j,k)$. 
Furthermore, the order of {these jobs'} tasks in the priority queue provides a scheme to charge the cost of the tasks regarding $j$ to the cost of tasks regarding $k$.

\smallskip

Figure~\ref{fig:illustration} shows how the charging is done.
Each row in the subfigures is a permutation of how the tasks regarding job $j$ and $k$ are executed.
The gray objects are tasks regarding $k$, and the white objects are tasks regarding $j$. 
The circles, rectangles, and rectangles with the wavy top are testing tasks, execution tasks, and execution-untested tasks, respectively.
The horizontal lines present the values of $c(k,j)$ (light blue) and $c(j,k)$ (dark blue).
The red arrows indicate how the cost of a task regarding $j$ is charged to that of a task regarding $k$ according to the order of the tasks in the priority queue. 
The charging $c(k,j)+c(j,k)$ to the cost of tasks regarding $k$ results in Lemmas~\ref{lem:N} and~\ref{lem:T}.
\shortversion{For the sake of space, the proof is provided in the full paper.}

\begin{figure}[H]
\centering
\begin{subfigure}{.4\textwidth}
  \includegraphics[width=\textwidth]{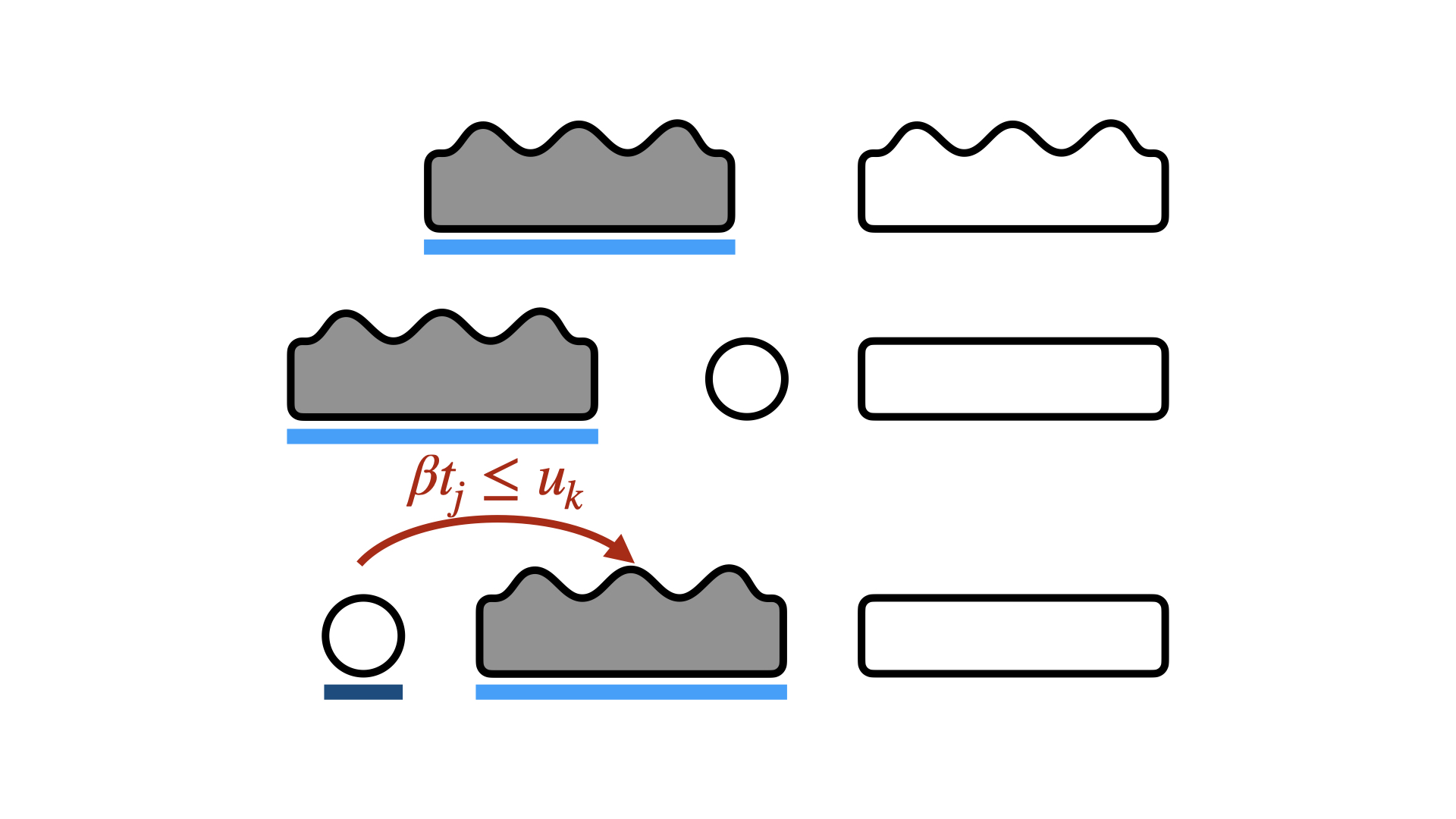}
  \caption{$k$ is not tested and $c(k,j) = u_k$~}
  \label{subfig:N1}
\end{subfigure}%
\begin{subfigure}{.4\textwidth}
  \includegraphics[width=\textwidth]{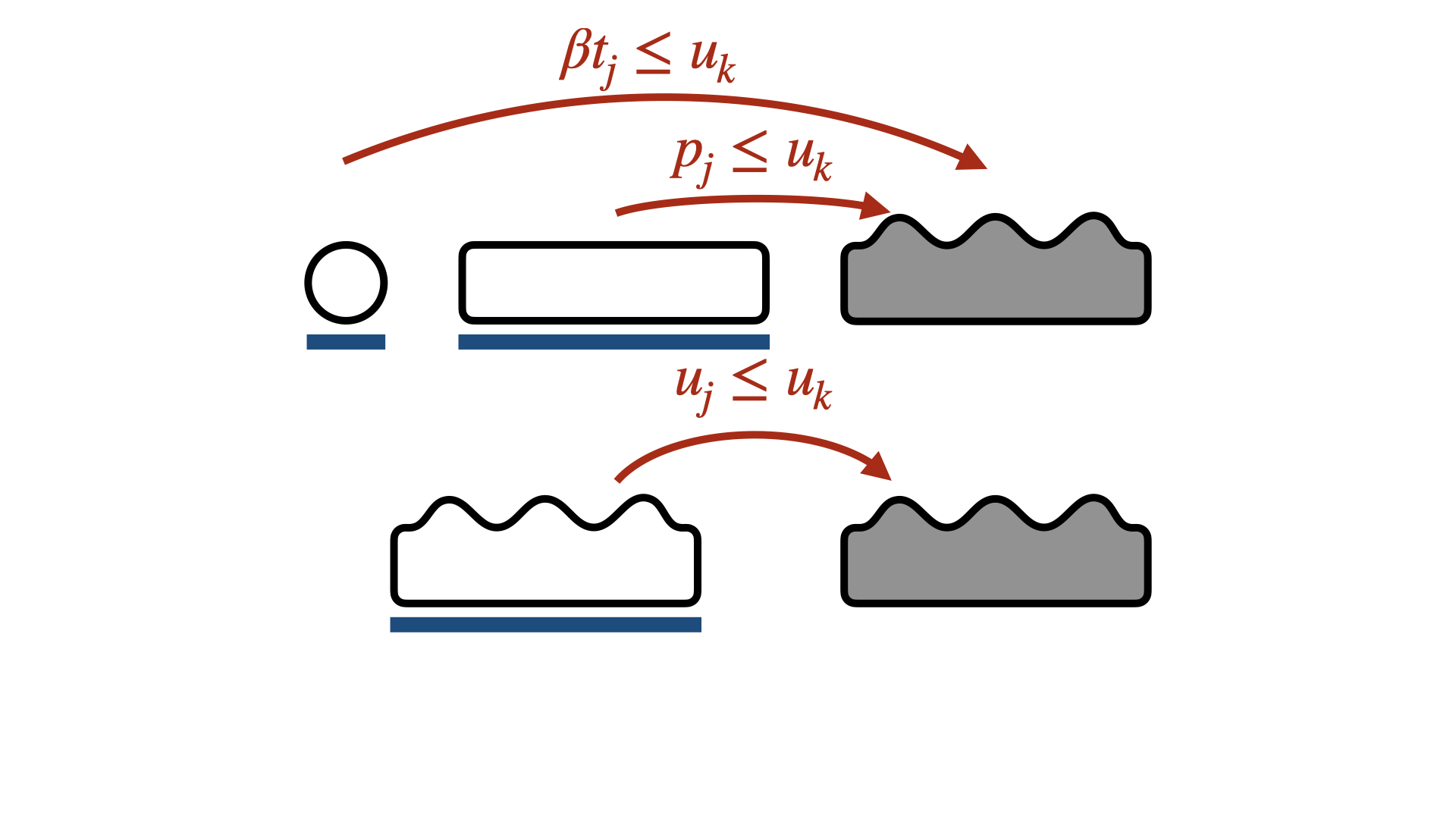}
  \caption{$k$ is not tested and $c(k,j) = 0$}
  \label{subfig:N2}
\end{subfigure}
\begin{subfigure}{.4\textwidth}
  \includegraphics[width=\textwidth]{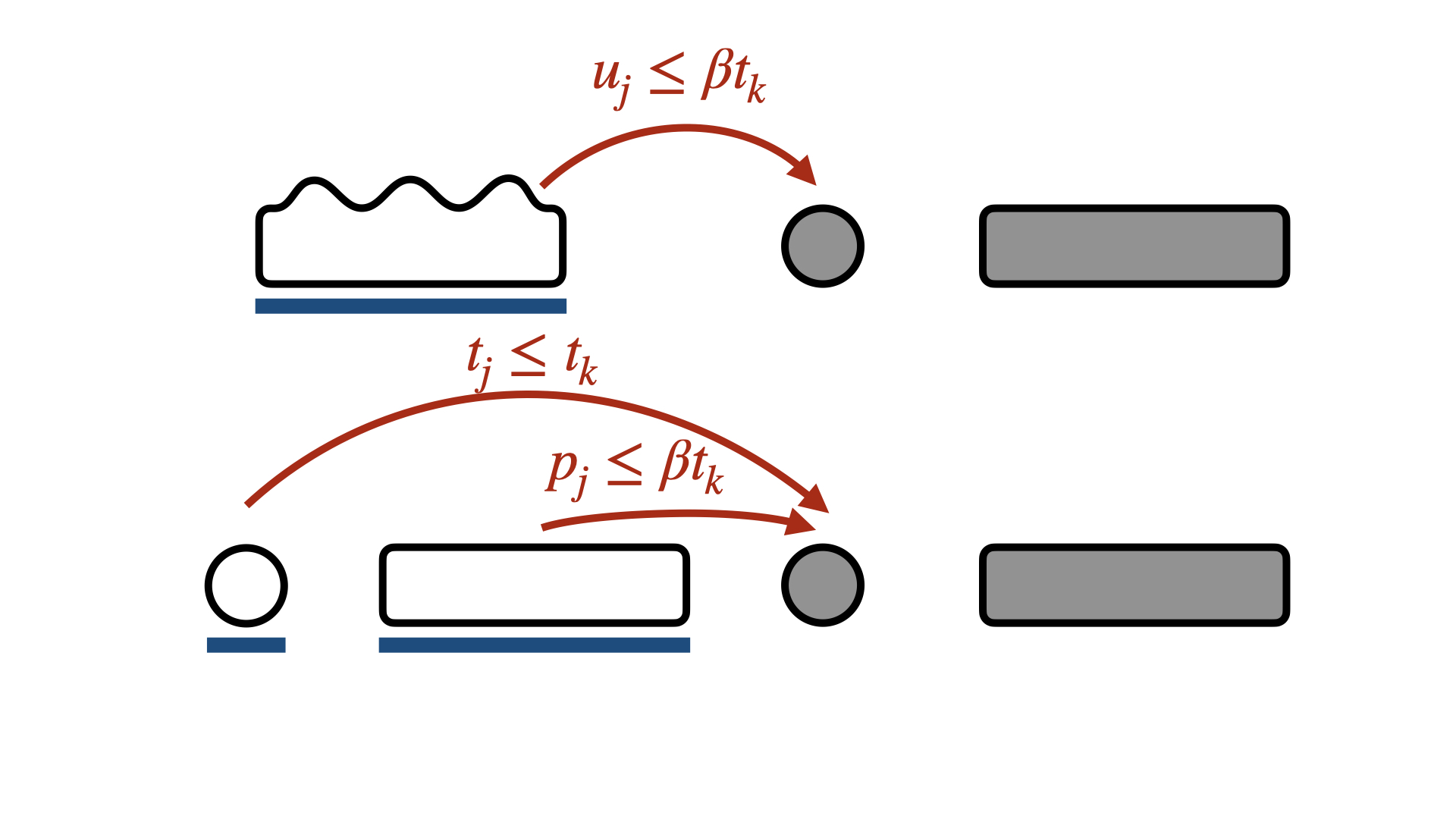}
  \caption{$k$ is tested and $c(k,j) = 0$ }
  \label{subfig:T1}
\end{subfigure}%
\begin{subfigure}{.4\linewidth}
  \includegraphics[width=\textwidth]{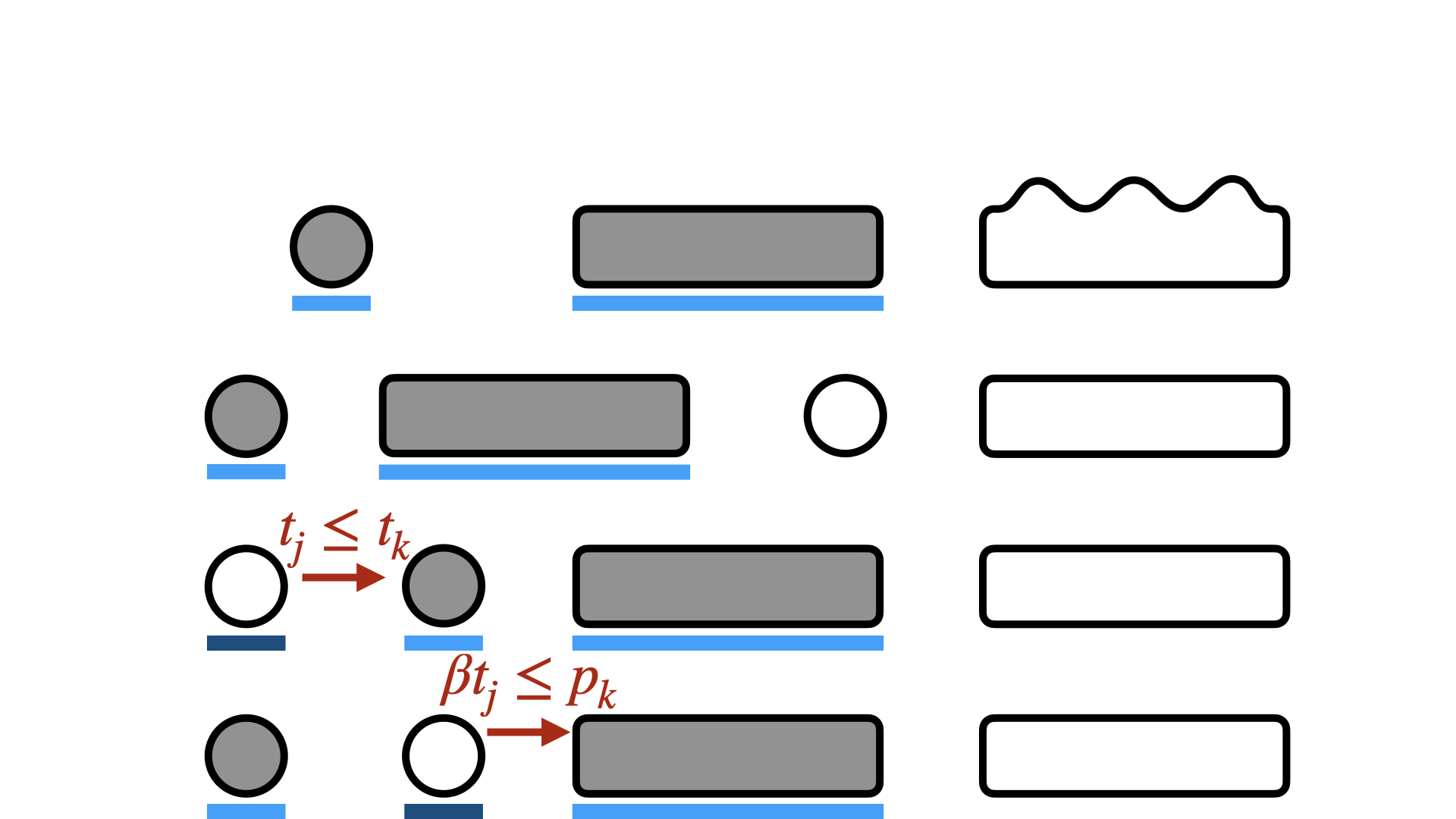}
  \caption{$k$ is tested and $c(k,j) = t_k+p_k$}
  \label{subfig:T2}
\end{subfigure}
\begin{subfigure}{.4\textwidth}
  \includegraphics[width=\textwidth]{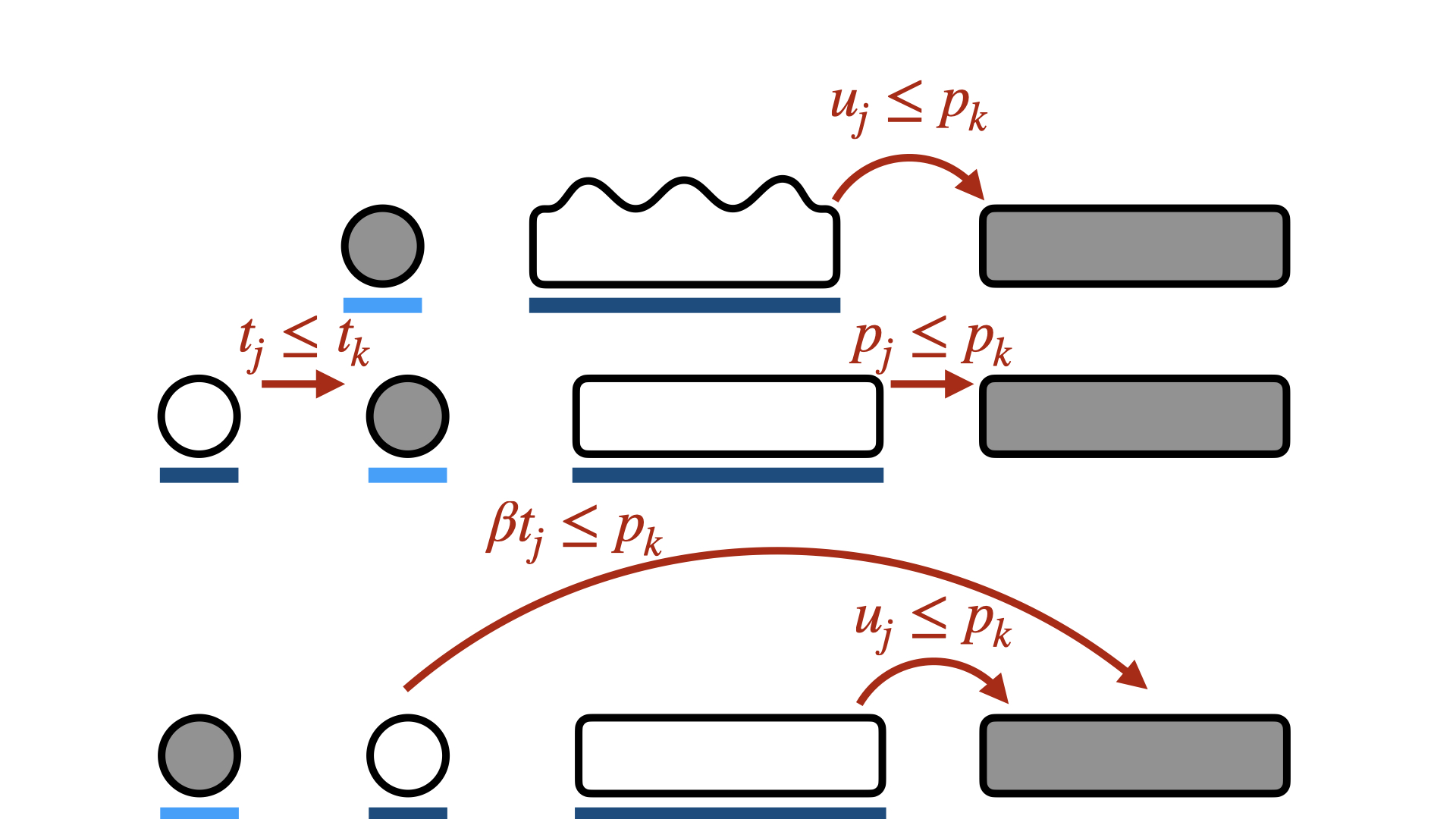}
  \caption{$k$ is tested and $c(k,j) = t_k$}
  \label{subfig:T3}
\end{subfigure}
\caption{The red arrows illustrate how to charge $c(k,j)+c(j,k)$ to the cost of tasks regarding $k$. Each row in the sub-figures is a permutation of how the tasks are executed. 
The circles and rectangles are testing tasks and execution tasks after testing, respectively. The rectangles with curly tops are execution tasks without testing. The tasks in gray are from the job $k$, and the tasks in white are from the job $j$.
The light blue and dark blue line segments under the tasks represent the contribution $c(k,j)$ and $c(j,k)$, respectively.
}
\label{fig:illustration}
\end{figure}

\hide{\begin{lemma}\cite{DBLP:conf/waoa/AlbersE20}
    If job $k$ is tested by $\AEalgo$, $\alg{p_k} \leq (1+\frac{1}{\alpha})\cdot\opt{p_k}$.
    If job $k$ is not tested by $\AEalgo$, $\alg{p_k} \leq {\alpha}\cdot\opt{p_k}$.
    For any job $k$, $\opt{p_k} \geq \min\{t_k, u_k\}$.
\end{lemma}}

\begin{lemma}
    \label{lem:N}
    If $\AEalgo$ does not test job $k$, 
    \[c(k,j)+c(j,k) \leq (1+\frac{1}{\beta})u_k.\]
\end{lemma}
\fullversion{
\begin{proof}
    Since job $k$ is not tested, $c(k,j) = u_k$ or $c(k,j) = 0$.
    The case $c(k,j) = u_k$ only happens when job $k$ is executed untested before $j$ is finished. Thus, $c(j,k)\leq t_j$ (see Figure~\ref{subfig:N1}). In this case, $t_j$ is executed before $u_k$ because $\beta t_j \leq u_k$. 
    Overall, $c(k,j)+c(j,k) \leq (1+\frac{1}{\beta})u_k$
    
    If $c(k,j) = 0$, all the tasks regarding job $j$ are done by $u_k$ (see Figure~\ref{subfig:N2}). Hence, $c(j,k) \leq \max\{u_j, t_j + p_j\}$. In the first case, $j$ is executed untested before $k$ because $u_j \leq u_k$. In the second case, both the (weighted) testing and processing time of job $j$ are less than $u_k$.
    Therefore, in the worst case, $c(k,j)+c(j,k) \leq (1+\frac{1}{\beta})u_k$.
    \qedhere
\end{proof}
}

\begin{lemma}
    \label{lem:T}
    If $\AEalgo$ tests job $k$, 
    \[c(k,j)+c(j,k) \leq \max\{2t_k + p_k, (1+\beta) t_k, t_k + (1+\frac{1}{\beta})p_k\}.\]
\end{lemma}
\fullversion{
\begin{proof}
    Since job $k$ is tested, $c(k,j) = 0$, $t_k+p_k$, or $t_k$.
    The case $c(k,j) = 0$ happens only when all tasks regarding job $j$ are done before testing $k$ (see Figure~\ref{subfig:T1}). 
    Therefore, $c(j,k) = \alg{p_j}$, which is $u_j$ or $t_j+p_j$. In the first case, $u_j\leq \beta t_k$. In the second case, $t_j\leq t_k$, and $p_j \leq \beta t_k$. Overall, $c(k,j)+c(j,k) \leq (1+\beta)t_k$ in this case.
    
    The case $c(k,j) = t_k+p_k$ happens only when $p_k$ is finished before the last task regarding job $j$ (see Figure~\ref{subfig:T2}). Therefore, $c(j,k) \leq t_j$. In this case, $t_j\leq \max\{t_k, \frac{p_k}{\beta}\}$. Overall, $c(k,j)+c(j,k) \leq \max\{2t_k+p_k, t_k+(1+\frac{1}{\beta})p_k\}$.

    If $c(k,j) = t_k$, $c(j,k) \leq \alg{p_j}$ since all tasks regarding job $j$ finish before the last task regarding job $k$ (see Figure~\ref{subfig:T3}). If $c(j,k) = u_j$, the execution of $j$ finished before $k$ because $u_j \leq p_k$. If $\alg{p_j} = t_j + p_j$, $t_j \leq \max\{t_k, \frac{p_k}{\beta}\}$, and $p_j \leq p_k$.
    Overall, $c(k,j)+c(j,k) \leq \max\{2t_k+p_k, t_k+(1+\frac{1}{\beta})p_k\}$.
    \qedhere
\end{proof}
}

Now, we can bound the competitive ratio of the $\AEalgo$ (Theorem~\ref{Thm:amortized}).
The idea is, depending on whether job $k$ is tested or not by the optimal schedule, the expressions in Lemmas~\ref{lem:N} and~\ref{lem:T} can be written as a function of $\alpha$, $\beta$, and $\opt{p_k}$. 
By selecting the values of $\alpha$ and $\beta$ carefully, we can balance the worst case ratio in the scenario where $k$ is executed-untested by the algorithm (Lemma~\ref{lem:N}) and that in the scenario where $k$ is tested by the algorithm (Lemma~\ref{lem:T}).

\begin{theorem}
    \label{Thm:amortized}
    The competitive ratio of $\AEalgo$ is at most 
    \begin{equation}
        \label{eq:amortized}
        \max\{\alpha(1+\frac{1}{\beta}), 1+\frac{1}{\alpha}+\frac{1}{\beta}, 1+ \beta, 2 , 1+\frac{2}{\alpha}\}
    \end{equation}
    \shortversion{By choosing $\alpha=\beta=\sqrt{2}$, $\AEalgo$ algorithm is $(1+\sqrt{2})$-competitive.
    The choice is optimal {for expression~(\eqref{eq:amortized})}.
    }
\end{theorem}

\fullversion{
\begin{proof}
    For simplicity, denote $r(\alpha, \beta)$ as $\max\{(1+\frac{1}{\beta})u_k, 2t_k + p_k, (1+\beta) t_k, t_k + (1+\frac{1}{\beta})p_k\}$.
    Recall that the $\AEalgo$ only tests job $k$ when $u_k\geq \alpha t_k$, and $p_k\leq u_k$ for all jobs $k$. 
    The argument is divided into two cases according to $\opt{p_k}$.
    
    \runtitle{Case 1: $\opt{p_k} = t_k+p_k$.}
    In this case,
    \begin{align*}
        \frac{r(\alpha,\beta)}{\opt{p_k}} &\leq \max\left\{\frac{(1+\frac{1}{\beta})u_k}{t_k+p_k}, 2, 1+\beta, 1+\frac{1}{\beta}\right\} \\
        &\leq \max\left\{(1+\frac{1}{\beta})\alpha, 2, 1+\beta, 1+\frac{1}{\beta}\right\}
    \end{align*} 

    \smallskip
    
    \runtitle{Case 2: $\opt{p_k} = u_k$.}
    In this case, 
    \begin{align*}
        \frac{r(\alpha,\beta)}{\opt{p_k}} &\leq \max\left\{ 1 + \frac{1}{\beta}, \frac{2t_k + p_k}{u_k}, \frac{(1+\beta) t_k}{u_k}, \frac{t_k + (1 + \frac{1}{\beta}) p_k}{u_k} \right\} \\
        &\leq \max \left\{ 1 + \frac{1}{\beta}, \frac{2}{\alpha} + 1, \frac{1 + \beta}{\alpha}, \frac{1}{\alpha} + 1 + \frac{1}{\beta} \right\}
    \end{align*} 
    \qedhere
\end{proof}
}

Note that by Theorem~\ref{Thm:amortized}, the $\AEalgo$ algorithm is $3$-competitive when $\alpha=\beta=1$, which matches the observation in Figure~\ref{fig:observation}. 

\smallskip

Our analysis framework provides room for adjusting the values of $\alpha$ and $\beta$. 
By selecting the values of $\alpha$ and $\beta$, we can tune the cost of tasks regarding $k$ that is charged. 
By selecting a value of $\alpha$ other than $1$, we can balance the penalty of making a wrong decision on testing a job or not. 
The capability of selecting a value of $\beta$ other than $1$ allows us to prioritize the testing tasks (which are scaled by $\beta$) and the execution tasks (which immediately decide a completion time of a job).
Finally, the performance of the algorithm is tuned by finding the best values of $\alpha$ and $\beta$. 

\smallskip

\fullversion{
\begin{corollary}
    By choosing $\alpha=\beta=\sqrt{2}$, $\AEalgo$ algorithm is $(1+\sqrt{2})$-competitive.
    This choice of $\alpha$ and $\beta$ is optimal for expression~\eqref{eq:amortized}.
\end{corollary}

\begin{proof}
    By setting $\alpha=\beta=\sqrt{2}$, all $\alpha(1+\frac{1}{\beta}), 1+\frac{1}{\alpha}+\frac{1}{\beta}, 1+ \beta$, and $1+\frac{2}{\alpha}$ are equal to $1+\sqrt{2}$.
    That is, the choice minimizes $\max\{\alpha(1+\frac{1}{\beta}), 1+\frac{1}{\alpha}+\frac{1}{\beta}, 1+ \beta, 2 , 1+\frac{2}{\alpha}\}$.
    Thus the corollary follows.
    \qedhere
\end{proof}
}

However, recall that the parameter $\alpha$ encodes the penalty for making a wrong guess on testing a job or not. 
When $\alpha = \sqrt{2}$, the penalty for testing a job we should not test is more expensive than that for executing-untested a job that we should test.
It inspires us to improve the algorithm further.

%% file: NewAlgo.tex
\subsection{An improved algorithm $\ouralgo$}
\label{subsec:improved}

Surprisingly, the introduction of amortization even sheds light on further improvement of the algorithm $\AEalgo$. 
We propose a new algorithm, called \emph{Prioritizing-Certain-Processing-time} (abbreviated as $\ouralgo$). 
The main difference between the two algorithms 
is that in the $\ouralgo$ algorithm after a job $j$ is tested, an execution task with weight $t_j+p_j$ is inserted into the queue instead of $p_j$ (see Algorithm~\ref{Alg:new_queue}).\footnote{The idea of using $t_j+p_j$ as the weight is also used in~\cite{DBLP:journals/algorithmica/GongCH23}.}
Intuitively, we prioritize a job by its certain (total) processing time $\alg{p_j}$, which can be $t_j+p_j$ or $u_j$.
Then, we can charge the total cost of tasks regarding a wrong-ordered $j$ to $\beta t_k$ or $\alg{p_k}$ all at once.



\hide{
\begin{algorithm}
    \caption{$\ouralgo$ algorithm}
    \label{Alg:ourALG}
    \begin{algorithmic}
        \State Initialize a priority queue $Q$
        \For{$j = 1, 2, 3, \cdots, n$}
            \If{$u_j\geq \alpha\cdot t_j$}
                \State Insert $\beta\cdot t_j$ into $Q$
            \Else
                \State Insert $u_j$ into $Q$
            \EndIf
        \EndFor
    \State \textbf{Updated Queue-Execution}$(Q)$
    \end{algorithmic}
\end{algorithm}
}

\begin{algorithm}[t]
\begin{algorithmic}
    \Procedure{Updated Queue-Execution}{$Q$}
        \While{$Q$ is not empty}
            \State {$x \gets$ Extract the smallest-weight} task in $Q$
            \If{{$x$ is a testing task for a job $j$}}
                \State {Test job $j$}
                \Comment{It takes $t_j$ time}
                \State {Insert an execution task with weight $t_j+p_j$ into $Q$}
            \ElsIf{{$x$ is an execution task for a job $j$}}
                \State Execute (tested) job $j$
                \Comment{It takes $p_j$ time}
            \Else \Comment{{$x$ is an execution-untested task for a job $j$}}
                \State Execute job $j$ untested
                \Comment{It takes $u_j$ time}
            \EndIf
        \EndWhile
    \EndProcedure
\caption{Procedure \textbf{Updated Queue-Execution} $(Q)$}
\label{Alg:new_queue}
\end{algorithmic}
\end{algorithm}

The new algorithm {$\ouralgo$ (Algorithm~\ref{Alg:Susanne} combined with Algorithm~\ref{Alg:new_queue})} has an improved estimation of $c(k,j)+c(j,k)$ when $c(j,k) = t_j+p_j$. 
However, when there is only one task regarding $j$ contributing to $c(j,k)$, the estimation of $c(k,j)+c(j,k)$ may increase.
Formally, we have the following two lemmas.

\begin{lemma}
    \label{lem:newN}
    Given two jobs $k\leq_o j$, if $\ouralgo$ does not test job $k$, 
    \[c(k,j)+c(j,k) \leq (1+\frac{1}{\beta})u_k.\]
\end{lemma}

\fullversion{
\begin{proof}
    Since job $k$ is not tested, $c(k,j) = u_k$ or $c(k,j) = 0$.
    The case $c(k,j) = u_k$ only happens when job $k$ is executed untested before $j$ is finished. Thus, $c(j,k)\leq t_j$. In this case, $t_j$ is executed before $u_k$ because $\beta t_j \leq u_k$. 
    Overall, $c(k,j)+c(j,k) \leq (1+\frac{1}{\beta})u_k$
    
    If $c(k,j) = 0$, all the tasks regarding job $j$ are finished before $u_k$.
    Hence, $c(j,k) \leq \max\{u_j, t_j + p_j\}$.
    In the first case, $j$ is executed untested before $k$ because $u_j \leq u_k$.
    In the second case, $t_j+p_j\leq u_k$.
    Therefore, in the worst case, $c(k,j)+c(j,k) \leq (1+\frac{1}{\beta})u_k$.
    \qedhere
\end{proof}
}

\begin{lemma}
    \label{lem:newT}
    Given two jobs $k\leq_o j$, if $\ouralgo$ tests job $k$, 
    \[c(k,j)+c(j,k) \leq \max\{2t_k + p_k, \beta t_k, (1+\frac{1}{\beta})(t_k+p_k)\}.\]
\end{lemma}

\fullversion{
\begin{proof}
    Since job $k$ is tested, $c(k,j) = 0$, $t_k+p_k$, or $t_k$.
    The case $c(k,j) = 0$ happens only when all tasks regarding job $j$ are done before testing $k$. 
    Therefore, $c(j,k) = \alg{p_j}$, which is $u_j$ or $t_j+p_j$. In the first case, $u_j\leq \beta t_k$. In the second case, $t_j+p_j\leq \beta t_k$. Overall, $c(k,j)+c(j,k) \leq \beta t_k$ in this case.
    
    The case $c(k,j) = t_k+p_k$ happens only when $p_k$ is finished before the last task regarding job $j$. Therefore, $c(j,k) \leq t_j$. In this case, $t_j\leq \max\{t_k, \frac{t_k+p_k}{\beta}\}$. Overall, $c(k,j)+c(j,k) \leq \max\{2t_k+p_k, t_k+(1+\frac{1}{\beta})(t_k+p_k)\}$.

    If $c(k,j) = t_k$, $c(j,k) \leq \alg{p_j}$ since all tasks regarding job $j$ finish before the last task regarding job $k$. No matter optimal schedule tests $j$ or not, $\alg{p_j} \leq t_k+p_k$. 
    Overall, $c(k,j)+c(j,k) \leq 2t_k+p_k$. 
    \qedhere
\end{proof}
}
\hide{
\begin{proof}
    The proof is similar to the proof of Lemma~\ref{lem:T}. 
    There are three cases: 1) $c(k,j) = 0$, 2) $c(k,j) = t_k+p_k$, and 3) $c(k,j) = t_k$.
    In the first case, $c(j,k) \leq \opt{p_j}$. Unlike the arguments in Lemma~\ref{lem:T}, now we can charge $t_j+p_j$ to $\beta t_k$. 

    In the second case, the worst-case scenario happens when $c(j,k) = t_j$ and $t_j>t_k$, where $t_j$ can only be charged to $\frac{t_k+p_k}{\beta}$. Overall, $c(k,j)+c(j,k) \leq (1+\frac{1}{\beta})(t_k+p_k)$. 

    In the third case, the worst-case scenario is that $c(j,k) = \opt{p_j}$. It happens because the last task regarding $j$ starts the execution of $k$. In this case, $c(j,k)$ can be charged to $t_k+p_k$. Overall, $c(k,j)+c(j,k) \leq 2t_k+p_k$.
    \qedhere
\end{proof}
}

Similar to the proof of Theorem~\ref{Thm:amortized}, we have the following competitiveness results of the $\ouralgo$ algorithm.

\begin{theorem}
    \label{Thm:improved}
    The competitive ratio of $\ouralgo$ is at most 
    \begin{equation}
        \label{eq:pcp}
        \max\{\alpha(1+\frac{1}{\beta}), 1+\frac{1}{\alpha}+\frac{1}{\beta} + \frac{1}{\alpha\beta}, \beta, 2 , 1+\frac{2}{\alpha}\} \enspace .
    \end{equation}
    \shortversion{By choosing $\alpha = \frac{1 + \sqrt{5}}{2}$ and $\beta = \frac{1 + \sqrt{5} + \sqrt{2 (7 + 5 \sqrt{5})}}{4}$,
    the competitive ratio of $\ouralgo$ is $\frac{1 + \sqrt{5} + \sqrt{2 (7 + 5 \sqrt{5})}}{4} \leq 2.316513$.
    The choice is optimal for expression~\eqref{eq:pcp}.
    }
\end{theorem}

\fullversion{
\begin{corollary}
    \label{Cor:improved}
    By choosing $\alpha = \phi = \frac{1 + \sqrt{5}}{2}$ and $\beta = \frac{\phi + \sqrt{5\phi + 1}}{2}$,
    the competitive ratio of $\ouralgo$ is $\frac{\phi + \sqrt{5\phi + 1}}{2}\leq 2.316513$.
    This choice of $\alpha$ and $\beta$ is optimal for expression~\eqref{eq:pcp}.
\end{corollary}

\begin{proof}
    We consider $\alpha(1+\frac{1}{\beta})$, $1+\frac{1}{\alpha}+\frac{1}{\beta} + \frac{1}{\alpha\beta}$, and $\beta$ in the $\max\{\ldots\}$.
    Let $S$ be the set of these three items.
    By making the first two items of $S$ equal, we obtain $\alpha = \frac{1 + \sqrt{5}}{2}$.
    We then make the second and the third items of $S$ equal, and with the obtained value of $\alpha$ plugged in, we get $\beta = \frac{\phi + \sqrt{5\phi + 1}}{2}$.
    With these specified values of $\alpha$ and $\beta$, the items in $S$ are all equal, and they are the largest items in the $\max\{\ldots\}$.
    Since $\alpha$ and $\beta$ are located in both numerators and denominators in $S$, any changes to $\alpha$ and $\beta$ yield one of the items larger.
    Thus no other choices of $\alpha$ and $\beta$ can provide a smaller value of the $\max\{\ldots\}$.
    \qedhere
\end{proof}
}

The selection of golden ratio for $\alpha$ balances the penalty of making a wrong guess for testing a job or not.

\smallskip

Note that using the analysis proposed in the work of Albers and Eckl~\cite{DBLP:conf/waoa/AlbersE20} on the new algorithm that put $t_j+p_j$ back to the priority list after testing job $j$, the competitive ratio is $\max\{\alpha, 1+\frac{1}{\alpha}\} + \max\{\alpha, 1+\frac{1}{\alpha}, \beta\}$. 
The best choice of the values is $\alpha = \phi$ and $\beta\in[1,\phi]$, and the competitive ratio is at most $2\phi$.

\subsection{Tightness of our analysis}\label{sec:algorithm_lower_bound}
In this section, we show that our analysis is close to tight.
Recall that the competitive ratio of $\ouralgo$ is at most $\max\{\alpha(1+\frac{1}{\beta}), 1+\frac{1}{\alpha}+\frac{1}{\beta} + \frac{1}{\alpha\beta}, \beta, 2 , 1+\frac{2}{\alpha}\}$. 
When $\alpha = \phi$ and $\beta = \frac{\phi + \sqrt{5\phi + 1}}{2}$, the ratio is dominated by the terms $\alpha(1+\frac{1}{\beta})$, $1+\frac{1}{\alpha}+\frac{1}{\beta} + \frac{1}{\alpha\beta}$, and $\beta$.
These terms correspond to different configurations of some pair of jobs $k$ and $j$.
For example, the bounds $\alpha(1+\frac{1}{\beta})$ and $\beta$ happen when at least one job $k$ or $j$ is executed untested by $\ouralgo$, and some properties must hold between the length of the tasks. 
On the other hand, the term $1 + \frac{1}{\alpha} + \frac{1}{\beta} + \frac{1}{\alpha\beta}$ happens when job $k$ is executed before job $j$ without testing, while $\ouralgo$ tests both jobs and the following conditions hold: 1) $t_k \leq t_j$, 2) $\beta \cdot t_j \leq t_k+p_k$, and 3) $t_k+p_k \leq t_j + p_j$. 
More specifically, to make the bound tightly holds between jobs $k$ and $j$, it must be $p_k = u_k$ and $t_j = \frac{t_j+p_j}{\beta}$.

These ``tight worst cases'' cannot occur between any pair of jobs. 
Consider the case with three jobs $1$, $2$, and $3$, where $t_1 = 1$, and $p_i = u_i = \alpha t_i$ for all $i \in \{1, 2, 3\}$.
In the following, we demonstrate that if jobs $1$ and $2$ form a pair of jobs that admits the bound $1+\frac{1}{\alpha}+\frac{1}{\beta}+\frac{1}{\alpha\beta}$ tightly, and so do jobs $2$ and $3$, then jobs $1$ and $3$ must not form such a bad configuration.

According to previous discussion, jobs $1$ and $2$ admit the bound $1+\frac{1}{\alpha}+\frac{1}{\beta}+\frac{1}{\alpha\beta}$ tightly when $p_1 = u_1 = \phi \cdot t_1 = \phi$, and $t_2 = \frac{t_1+p_1}{(\phi+\sqrt{5\phi+1})/2} \approx 1.130$.
Similarly, jobs $2$ and $3$ admit the bound tightly when $p_2 = u_2 = \phi \cdot t_2 \approx 2.959$, and $t_3 = \frac{t_2+p_2}{(\phi+\sqrt{5\phi+1})/2} \approx 1.278$.
Now, consider the pair of jobs $1$ and $3$. 
Although $t_1 \leq t_3$ and $t_1+p_1 \leq t_3 + p_3$, the critical condition $\beta \cdot t_3 \leq t_1+p_1$ does not hold. 
Therefore, by $\ouralgo$, jobs $1$ and $3$ actually form a sub-schedule where $t_1$ and $p_1$ are executed before $t_3$ and $p_3$, which only admits an upper bound of $\phi$.


To provide a lower bound of $\ouralgo$ with our choice of $\beta = \frac{\phi + \sqrt{5\phi + 1}}{4}$ and $\alpha = \phi$, we design an adversary instance, where each pair of jobs \emph{loosely} admits the bound of $1+\frac{1}{\alpha}+\frac{1}{\beta}+\frac{1}{\alpha\beta}$.

\begin{theorem}
    \label{Thm:ouralgo_lb}
    Let $\alpha = \phi$ and $\beta = \frac{\phi + \sqrt{5\phi + 1}}{4}$, the competitive-ratio of $\ouralgo$ is at least $2.26177$.
\end{theorem}

\begin{proof}
    We consider the case where for each job $j$, $u_j \geq \alpha$, and for each $k \leq_o j$
    \begin{itemize}
        \item $t_k \leq t_j$
        \item $t_k + p_k \geq \beta \cdot t_j$
        \item $t_k + p_k \leq t_j + p_j$
    \end{itemize}
    Then, the algorithm first tests all jobs before executing any of the tested jobs.
    If for each pair of jobs these inequalities must hold, only for the pair $1$ and $n$ can all inequalities be tight.
    The adversary constructs the input as follows: for all jobs $j$, $t_j = 1 + \frac{j-1}{n-1}\cdot\left(\frac{1 +\alpha}{\beta}-1\right)$, and $u_j = p_j = \alpha \cdot t_j$.
    Note that the testing time is an increasing arithmetic sequence, hence $t_k \leq t_j$ and $t_k + p_k \leq t_j + p_j$ hold trivially.
    Furthermore, $t_1 + p_1 = 1 + \alpha = \beta \cdot \frac{1+\alpha}{\beta} = \beta \cdot \left(1 + \frac{n-1}{n-1} \cdot \left( \frac{1+\alpha}{\beta} - 1 \right) \right) = \beta t_n$, thus $t_k + p_k \geq \beta t_j$ holds for all pairs $k\leq_oj$ while the inequality is tight for pair $1,n$.

    As for any job $j$, $t_j + p_j > u_j$, the optimal algorithm will test no jobs and execute them untested for cost $u_j = p_j =\alpha t_j$.
    Then, with $\alpha = \phi$ and $\beta = \frac{\phi + \sqrt{5\phi + 1}}{4}$, it follows for $n$ tending towards infinity that
{
\allowdisplaybreaks 
    \begin{align*}
    \lim\limits_{n\to\infty} \frac{\cost(\ouralgo)}{\cost(\optimal)} &= \lim\limits_{n\to\infty} \frac{\sum\limits_{j=1}^n\left( \sum\limits_{i=1}^n t_i + \sum\limits_{i\leq_o j} p_i\right)}{\sum\limits_{j=1}^n\sum\limits_{i\leq_o j} p_i} \\
    &= \lim\limits_{n\to\infty} \frac{\sum\limits_{j=1}^n\left( \sum\limits_{i=1}^n t_i + \sum\limits_{i \leq_o j} \phi t_i\right)}{\sum\limits_{j=1}^n\sum\limits_{i\leq_oj} \phi t_i} \\
    &= \lim\limits_{n\to\infty} \frac{n \sum\limits_{i=1}^n t_i}{\phi\sum\limits_{j=1}^n\sum\limits_{i\leq_oj} t_i} + \frac{\phi\sum\limits_{j=1}^n\sum\limits_{i \leq_o j} t_i}{\phi\sum\limits_{j=1}^n\sum\limits_{i\leq_oj} t_i}\\
    &= 1 + \lim\limits_{n\to\infty} \frac{n \sum\limits_{i=1}^n \left( 1 + \frac{i-1}{n-1}\cdot \frac{\phi^2-\beta}{\beta}\right)}{\phi\sum\limits_{j=1}^n\sum\limits_{i\leq_oj} \left( 1 + \frac{i-1}{n-1}\cdot \frac{\phi^2-\beta}{\beta} \right)}\\
    &= 1 + \lim\limits_{n\to\infty} \frac{n\left( \sum\limits_{i=1}^n  1 + \frac{\phi^2-\beta}{\beta (n-1)}\cdot \sum\limits_{i=1}^n (i - 1)\right)}{\phi\sum\limits_{j=1}^n \left( \sum\limits_{i\leq_oj}  1 + \frac{\phi^2-\beta}{\beta (n-1)}\cdot \sum\limits_{i\leq_oj} (i - 1) \right)}\\
    &= 1 + \lim\limits_{n\to\infty} \frac{n\left( n + \frac{\phi^2-\beta}{\beta (n-1)}\cdot \frac{(n-1)n}{2}\right)}{\phi\sum\limits_{j=1}^n \left( (j-1) + \frac{\phi^2-\beta}{\beta (n-1)}\cdot \frac{(j-2)(j-1)}{2} \right)}\\
    &= 1 + \lim\limits_{n\to\infty} \frac{n\left(n + \frac{\phi^2-\beta}{2\beta}\cdot n\right)}{\phi\left(\sum\limits_{j=1}^n  (j-1) + \frac{\phi^2-\beta}{2\beta (n-1)}\cdot \sum\limits_{j=1}^n \left(j-2\right)\left(j-1\right)\right)}\\
    &= 1 + \lim\limits_{n\to\infty} \frac{n^2 + \frac{\phi^2-\beta}{2\beta}\cdot n^2}{\phi\left(\frac{(n-1)n}{2} + \frac{\phi^2-\beta}{2\beta (n-1)}\cdot \frac{(n-2)(n-1)n}{3} \right)}\\
    &= 1 + \lim\limits_{n\to\infty} \frac{\left(1 + \frac{\phi^2-\beta}{2\beta}\right)n^2}{\frac{\phi}{2}\cdot(n^2 - n) + \frac{\phi^3-\phi\beta}{6\beta}\cdot (n^2-2n) }\\
    &\hspace{-22pt}\stackrel{\text{L'H\^{o}pital's rule}}{=} 1 + \lim\limits_{n\to\infty} \frac{\left(1 + \frac{\phi^2-\beta}{2\beta}\right)2}{\frac{\phi}{2} \cdot 2 + \frac{\phi^3-\phi\beta}{6\beta}\cdot 2 }\\
    &= 1 + \frac{\left(2 + \frac{\phi^2-\beta}{\beta}\right)}{\phi + \frac{\phi^3-\phi\beta}{3\beta} }\\
    &= 2.26177
\end{align*}
\qedhere
}
\end{proof}

%% file: Prmp.tex
\subsection{Preemption}
\label{subsec:preemption}

\hide{In this subsection, }We show that preempting the tasks does not improve the competitive ratio.
Intuitively, we show that given an algorithm $A$ that generates a preemptive schedule, we can find another algorithm $B$ that is capable of simulating $A$ and performs the necessary merging of preempted parts.
The simulation may make the timing of $A$'s schedule gain extra information about the real processing times earlier due to the advance of a testing task.
However, a non-trivial $A$ can only perform better by receiving the information earlier. 
Thus, $B$'s non-preemptive schedule has a total completion time at most that of $A$'s schedule.
\hide{
Intuitively, we show that given a schedule $S$ where jobs are preempted, we can transform it into a non-preemptive schedule that  performs the necessary merging of preempted parts.
The transformation may make the timing of $S$ gaining extra information about the real processing times earlier due to the advance of a testing task.
However, a non-trivial $S$ can only perform better by receiving the information earlier. 
Thus the schedule after transformation is non-preemptive and has a total completion time at most that of $S$.
}
\begin{lemma}
    In the \pr~problem on a single machine,
    if there is an algorithm that generates a preemptive schedule, then we can always find another algorithm that generates a non-preemptive schedule and performs as well as the previous algorithm in terms of competitive ratios.
\end{lemma}

\fullversion{
\begin{proof}
    \hide{
    Given an algorithm $A$ that generates a preemptive schedule, we can find another algorithm $B$ that is capable of simulating $A$ and performs the necessary merging of preempted parts.
    By postponing a preempted part of a job and combining it with the later part of the job, the total completion time does not increase.
    Although this operation may change the behavior of $A$ due to the advance of a testing task, a non-trivial $A$ can only perform better by receiving the information of the testing task earlier.
    $B$ can still simulate $A$ by carefully following $A$'s behavior.
    Thus the schedule generated by $B$ is non-preemptive and has a total completion time at most that of $A$.
    }
    In a preemptive schedule, a job may be divided into multiple in-contiguous parts $s_1, s_2, \cdots$.
    In order to obtain a corresponding non-preemptive schedule, one may need to reschedule these parts such that they are executed together.
    One way to do so is to \emph{right-merge} each of these parts.
    Consider two parts $s_i$ and $s_{i+1}$ and the sequence of tasks $S$ located in between $s_i$ and $s_{i+1}$.
    A right-merging of $s_i$ changes the subsequence of tasks from $(s_i, S, s_{i+1})$ to $(S, s_i, s_{i+1})$.
    This operation varies the total completion time non-increasingly, since only the completion times of the jobs corresponding to $S$ have been changed, and they cannot increase.
    By right-merging all the parts for each job, we can obtain a non-preemptive schedule with equal or smaller total completion time.
    In the following paragraphs, we prove that such right-merging is always possible.
    More precisely, we show that: given an algorithm $A$ that generates a preemptive schedule, we can find another algorithm $B$ that is capable of simulating $A$ and performs necessary right-merging.
    The schedule generated by $B$ is non-preemptive and has a total completion time
    at most
    that of $A$.
    Thus the lemma follows.
    
    Algorithm $B$ will simulate algorithm $A$ and perform a right-merging for each part of preempted tasks that $A$ generates.
    Algorithm $A$ may change its behavior based on the results of testing tasks, and $B$ must follow $A$'s behavior carefully.
    We elaborate on $B$'s behavior in the following two cases.
    For ease of analysis, we assume without loss of generality that $B$ only performs a right-merging of $s_i$ and leaves the other parts still  preempted.
    One can apply the arguments repeatedly and complete the proof.

    Suppose $S$ does not contain any testing tasks.
    Let $t$ be the last testing task located before $s_i$ and $t'$ be the first testing task located after $s_{i+1}$, which means the subsequence containing these tasks is $(t, \ldots, s_i, S, s_{i+1}, \ldots, t')$.
    At the moment immediately after $t$ is executed, $B$ is able to simulate $A$'s behavior between $t$ and $t'$, and thus $B$ can perform a right-merging of $s_i$ and execute the other tasks in between $t$ and $t'$ accordingly.
    On the other hand, suppose $S$ contains testing tasks.
    Let $t$ be the last testing task in $S$ and $t'$ be the first testing task located after $s_{i+1}$, which means the subsequence containing these tasks is $(s_i, S, s_{i+1}, ..., t')$ where $t \in S$.
    In order to perform a right-merging of $s_i$, $B$ needs to postpone the execution of $s_i$.
    This makes the testing tasks in $S$ being executed earlier, and $A$ may change its behavior and become another algorithm, denoted by $A'$, due to the advance of the test results.
    We note that, for any non-trivial $A'$, it performs no worse than $A$ with respect to the tasks located before $t$ (excepting the postponed $s_i$).
    The detailed behavior of $B$ is as follows.
    It first postpones $s_i$ and sees if $A$ changes its behavior.
    If $A$ does not change its behavior, $B$ simulates $A$ until $t$.
    Otherwise, $B$ simulates $A'$ until $t$.
    Algorithm $B$ also executes the tasks assigned in the simulation (except $s_i$).
    After that, $B$ performs a right-merging of $s_i$ and executes the other tasks in between $t$ and $t'$ accordingly.
    \qedhere
\end{proof}
}

%% file: Randomized.tex
The amortization also helps improve the performance of randomized algorithms.
We combine the $\ouralgo$ algorithm with the framework in the work of Albers and Eckl~\cite{DBLP:conf/waoa/AlbersE20}, where instead of using a fixed threshold $\alpha$, a job $j$ is tested with probability $\prob_j$, which is a function of $u_j$, $t_j$, and $\beta$.

\runtitle{Our randomized algorithm $\randalgo$}. For any job $j$ with $\frac{u_j}{t_j} < 1$ or $\frac{u_j}{t_j}>3$, we insert $u_j$ or $\beta t_j$ into the queue, respectively.  
For any job $j$ with $1\leq\frac{u_j}{t_j}\leq3$, we insert $\beta t_j$ into the queue with probability $\prob_j$ and insert $u_j$ with probability $1 - \prob_j$.
\hide{\footnote{One can verify that $0 \leq \prob_j \leq 1$ when $1 \leq r_j \leq 3$.}}
Once a testing task $t_j$ is executed, we insert $t_j+p_j$ into the queue.
(See Algorithm~\ref{Alg:randomized}.)

\begin{algorithm}[t]
    \caption{$\randalgo$ algorithm}
    \label{Alg:randomized}
    \begin{algorithmic}
        \State Initialize a priority queue $Q$
        \For{$j = 1, 2, 3, \cdots, n$}
            \State Let $r_j \gets \frac{u_j}{t_j}$
            \If{$r_j < 1$}
                \State $\prob_j \gets 0$
            \ElsIf{$r_j > 3$}
                \State $\prob_j \gets 1$
            \Else
                \State $\prob_j= \frac{3r_j^2-3r_j}{3r_j^2-4r_j+3}$
            \EndIf
            
            \State {Choose one of $\beta t_j$ and $u_j$ randomly with probability $\prob_j$ for $\beta t_j$ and $1 - \prob_j$ for~$u_j$}
            \State {Insert a testing task with weight $\beta t_j$ into $Q$ if $\beta t_j$ is chosen, and insert an execution-untested task with weight $u_j$ into $Q$ otherwise}
        \EndFor
    \State \textbf{Updated Queue-Execution}$(Q)$ \Comment{See Algorithm~\ref{Alg:new_queue}}
    \end{algorithmic}
\end{algorithm}

\shortversion{
\runtitle{Analysis.} The following lemma can be proven using Lemma~\ref{lem:newN} and Lemma~\ref{lem:newT}.}

\begin{lemma}
    The expected total completion time of the $n$ jobs is at most 
    \[\sum_j\sum_{k\leq_o j} (1+\frac{1}{\beta})u_k(1-\prob_k) + \max\{2t_k + p_k, \beta t_k, (1+\frac{1}{\beta})(t_k+p_k)\}\prob_k,\]
    where $\prob_k$ is the probability that job $k$ is tested. 
\end{lemma}

\shortversion{
Depending on whether the jobs are tested or not in the optimal schedule, the expected total completion time can be expressed by functions with the variables: the probability, the parameters of the jobs, and $\beta$.
We design the probability $\prob_k$ by balancing the costs between the worst cases where $\opt{p_k} = u_k$ or $\opt{p_k} = t_k+p_k$. 
Note that there are cases where the ``ideal'' value of $\prob_k$ is outside the range $[0,1]$.
We take care of these special cases by setting $\prob_k$ as $0$ or $1$ if the ideal value is smaller than $0$ or larger than $1$, respectively.
}

\fullversion{
\begin{proof}
    By Lemma~\ref{lem:newN} and Lemma~\ref{lem:newT}, $\expect[c(k,j)+c(j,k)\mid k \text{ is not tested}] \leq (1+\frac{1}{\beta})u_k$, and $\expect[c(k,j)+c(j,k)\mid k \text{ is tested}] \leq \max\{2t_k + p_k, \beta t_k, (1+\frac{1}{\beta})(t_k+p_k)\}$. 
    Therefore, $\expect[c(k,j)+c(j,k)] = (1+\frac{1}{\beta})u_k(1-\prob_k) + \max\{2t_k + p_k, \beta t_k, (1+\frac{1}{\beta})(t_k+p_k)\}\prob_k$.
    \qedhere
\end{proof}
}

\begin{theorem}
    \label{Thm:randomized}
    Let $r_k$ denote $\frac{u_k}{t_k}$. 
    The expected competitive ratio of
    $\randalgo$
    is at most
    \[\max_k \frac{(1+\frac{1}{\beta})u_k(1-\prob_k)+\max\{2t_k+p_k, \beta t_k, (1+\frac{1}{\beta})(t_k+p_k)\}\prob_k}{\opt{p_k}}\text{, where }\]
    \[\prob_k = \frac{(\beta+1)(r_k-1)}{\beta(\max\{\frac{2}{r_k}+1, \frac{\beta}{r_k}, (1+\frac{1}{\beta})(1+\frac{1}{r_k})\}-\max\{2,\beta,1+\frac{1}{\beta}\}+r_k-1)+r_k-1}\] 
    if $r_k\in [1,3]$, $\prob_k = 0$ if $r_k < 1$, and $\prob_k = 1$ if $r_k>3$.
    \shortversion{
    By choosing $\beta = 2$, the ratio is $\frac{3 (7 + 3 \sqrt{6})}{20} \leq 2.152271$.
    The choice of $\beta$ is optimal.
    }
\end{theorem}

\fullversion{
\begin{proof}
    \begin{align*}
        &\expect[\sum_j C_j] = \expect[\sum_j \sum_{k\leq_o j}c(k,j)+c(j,k)] = \sum_j \sum_{k\leq_o j} \expect[c(k,j)+c(j,k)]\\
        & = \sum_j \sum_{k\leq_o j} \frac{(1+\frac{1}{\beta})u_k(1-\prob_k) + \max\{2t_k + p_k, \beta t_k, (1+\frac{1}{\beta})(t_k+p_k)\}\prob_k}{\opt{p_k}}\cdot\opt{p_k}\\
        & \leq \max_k \frac{(1+\frac{1}{\beta})u_k(1-\prob_k) + \max\{2t_k + p_k, \beta t_k, (1+\frac{1}{\beta})(t_k+p_k)\}\prob_k}{\opt{p_k}}\cdot \cost(\optimal)
    \end{align*}

    Next, we explain how to find $\prob_k$.
    There are two cases of $\opt{p_k}$: 1) $\opt{p_k} = u_k$, and 2) $\opt{p_k} = t_k+p_k$. 
    In the first case, the expected competitive ratio is at most 
    \begin{equation}
        \label{eq:rand_N}
        (1+\frac{1}{\beta})(1-\prob_k) + \max\{\frac{2}{r_k}+1, \frac{\beta}{r_k}, (1+\frac{1}{\beta})(1+\frac{1}{r_k})\}\prob_k
    \end{equation}
    Similarly, in the second case, the expected competitive ratio is at most
    \begin{equation}
        \label{eq:rand_T}
        (1+\frac{1}{\beta})r_k(1-\prob_k) + \max\{2,\beta,1+\frac{1}{\beta}\}\prob_k
    \end{equation}
    We optimize $\prob_k$ by letting {expressions}~\eqref{eq:rand_N} and~\eqref{eq:rand_T} equal.
    And $\prob_k$ is obtained as stated.
    \qedhere
\end{proof}
}

\hide{
\begin{corollary}
    By choosing $\beta = 2$,
    $\randalgo$
    has expected competitive ratio $\frac{3 (7 + 3 \sqrt{6})}{20} \leq 2.152271$.
    In this case, job $j$ is tested in a probability of $\frac{3r_j^2-3r_j}{3r_j^2-4r_j+3}$ for $1 \leq r_j \leq 3$, where $r_j = \frac{u_j}{t_j}$.
    The choice of $\beta$ is optimal.
\end{corollary}
}

\fullversion{
\begin{corollary}
    By choosing $\beta = 2$,
    $\randalgo$
    has expected competitive ratio $\frac{3 (7 + 3 \sqrt{6})}{20} \leq 2.152271$. In this case, job $j$ is tested in a probability of $\frac{3r_j^2-3r_j}{3r_j^2-4r_j+3}$ for $1 \leq r_j \leq 3$, where $r_j = \frac{u_j}{t_j}$.
\end{corollary}

\begin{proof}
    We show that for any $r_j$, both expressions~\eqref{eq:rand_N} and~\eqref{eq:rand_T} are at most $\frac{3 (7 + 3 \sqrt{6})}{20}$.
    If $r_j > 1$, $\prob_j = 0$ by the algorithm.
    Expression~\eqref{eq:rand_N} is $\frac{3}{2}$, and expression~\eqref{eq:rand_T} is $\frac{3}{2} r_j \leq \frac{3}{2}$.
    If $r_j < 3$, $\prob_j = 1$ by the algorithm.
    Expression~\eqref{eq:rand_N} is $\frac{3}{2} (1 + \frac{1}{r_j}) \leq \frac{3}{2} (1 + \frac{1}{3}) = 2$, and expression~\eqref{eq:rand_T} is $2$.
    Otherwise, $1 \leq r_j \leq 3$.
    To find the max of expressions~\eqref{eq:rand_N} and~\eqref{eq:rand_T}, we make the two expressions equal.
    This gives us $\prob_j = \frac{3 r_j^2 - 3 r_j}{3 r_j^2 - 4 r_j + 3}$.
    By plugging $\prob_j$ back to expression~\eqref{eq:rand_N}, we obtain $\frac{9 r_j^2 -3 r_j}{6 r_j^2 - 8 r_j + 6}$.
    This function has the maximum $\frac{3 (7 + 3 \sqrt{6})}{20}$, which happens at $r_j = 1 + \sqrt{\frac{2}{3}} \approx 1.816497$.
    \qedhere
\end{proof}
}

\fullversion{
\begin{lemma}
    \label{Lem:bestBeta}
    The choice of $\beta = 2$ minimizes the expected competitive ratio in the analysis among all possible $\beta > 0$.
\end{lemma}

\begin{proof}
    Given any $\beta > 0$, we show that there exists an $r_j$ such that the max of expressions~\eqref{eq:rand_N} and~\eqref{eq:rand_T} is at least $\frac{3 (7 + 3 \sqrt{6})}{20}$, and the lemma follows.
    To find the max of the two expressions, we use $\prob_j$ stated in Theorem~\ref{Thm:randomized}.
    For ease of analysis, we explicitly list the following two expressions that will be referred frequently in the proof.
    \begin{equation}
        \max\{\frac{2}{r_j} + 1, \frac{\beta}{r_j}, (1 + \frac{1}{\beta}) (1 + \frac{1}{r_j})\}
        \label{eq:x}
    \end{equation}
    \begin{equation}
        \max\{2, \beta, 1 + \frac{1}{\beta}\}
        \label{eq:y}
    \end{equation}
    We consider four cases of $\beta$.
    If $0 < \beta \leq 1$, we set $r_j = 2$.
    In this case, expression~\eqref{eq:x} is $\frac{3}{2} (1 + \frac{1}{\beta})$ and expression~\eqref{eq:y} is $1 + \frac{1}{\beta}$.
    The max of expressions~\eqref{eq:rand_N} and~\eqref{eq:rand_T} is $\frac{4}{3} + \frac{4}{3 \beta} \geq \frac{4}{3} + \frac{4}{3} \geq 2.66$.
    For $1 < \beta \leq 2$, we set $r_j = 1 + \sqrt{\frac{2}{3}}$.
    We have that expression~\eqref{eq:x} is $(1 + \frac{1}{\beta}) (4 - \sqrt{6})$, and expression~\eqref{eq:y} is $2$.
    The max of expressions~\eqref{eq:rand_N} and~\eqref{eq:rand_T} is $\frac{(1 + \beta) (\sqrt{6} \beta + \sqrt{6} + 6)}{2 \beta ((3 - \sqrt{6}) \beta - \sqrt{6} + 6)}$.
    This function decreases in the domain $1 < \beta \leq 2$ when $\beta$ goes larger.
    So, the function has the minimum in the domain at $\beta = 2$, which is $\frac{3 (7 + 3 \sqrt{6})}{20}$.
    If $2 < \beta \leq 2 + \sqrt{\frac{2}{3}}$, we also set $r_j = 1 + \sqrt{\frac{2}{3}}$.
    In this case, expression~\eqref{eq:x} is $(4 - \sqrt{6}) (1 + \frac{1}{\beta})$ and expression~\eqref{eq:y} is $\beta$.
    The max of expressions~\eqref{eq:rand_N} and~\eqref{eq:rand_T} is $\frac{(\beta + 1) (3 \beta^2 - (\sqrt{6} + 6) (\beta + 1))}{\beta (3 \beta^2 + (2 \sqrt{6} - 12) (\beta + 1))}$.
    This function increases in the domain $2 < \beta \leq 2 + \sqrt{\frac{2}{3}}$ when $\beta$ goes larger.
    So, the function has the minimum in the domain at $\beta = 2$, which is also $\frac{3 (7 + 3 \sqrt{6})}{20}$.
    Otherwise, $\beta > 2 + \sqrt{\frac{2}{3}}$.
    We set $r_j = \beta - 1$.
    We have that expression~\eqref{eq:x} is $\frac{\beta + 1}{\beta - 1}$, and expression~\eqref{eq:y} is $\beta$.
    The max of expressions~\eqref{eq:rand_N} and~\eqref{eq:rand_T} is $\frac{\beta^2 - 1}{2}$.
    This function increases in the domain $\beta \geq 2 + \sqrt{\frac{2}{3}}$ when $\beta$ goes larger.
    So, the function has the minimum in the domain at $\beta = 2 + \sqrt{\frac{2}{3}}$, which is $\frac{11 + 4 \sqrt{6}}{6} \geq 3.46$.
    \qedhere
\end{proof}
}

%% file: Parallel.tex
\subsection{Review of the framework in~\cite{DBLP:journals/algorithmica/GongCH23}}
\label{sec:para_algo}

For multiple machines, Gong et al.~\cite{DBLP:journals/algorithmica/GongCH23} proposed a framework based on Albers and Eckl's work.
When $m$ tends to infinity, the ratio of their algorithm is $2.92706$.
The authors employed the lower bound that the optimal cost is at least
\begin{equation} \label{eq:opt_LB_parallel}
    \max\{\sum_{j=1}^n p_j^*, \frac{1}{m}\sum_{j=1}^n \sum_{k=1}^j p_k^* + (\frac{1}{2}-\frac{1}{2m})\sum_{j=1}^n p_j^*\},\footnote{This is recasted for fitting our framework. And the order of $j$ is according to the execution order of the jobs in the optimal schedule.} 
\end{equation}
where $p_j^* = \min\{u_j, t_j+p_j\}$, and proposed a $2\phi$-competitive algorithm, adapted from $\AEalgo$, for variable testing time case and a $\phi+\frac{\phi+1}{2}\cdot(1-\frac{1}{m})$-competitive algorithm for uniform testing time case. 
The algorithms proposed by the authors were modified from the $\AEalgo$ with a more sophisticated ordering of the tasks (testing, execution, and execution-untested). More specifically, Gong et al. also use the priority list of the jobs where the jobs $j$ after testing have weights $t_j+p_j$.
In this section, we extend our technique to the framework for multiple machines proposed by Gong et al.

\bigskip

\runtitle{The algorithm for multiple machines.}
The algorithm in~\cite{DBLP:journals/algorithmica/GongCH23} is similar to $\ouralgo$ (Algorithm~\ref{Alg:Susanne}).
We recast the algorithm, as $\paralgo$, by replacing the subroutine described in Algorithm~\ref{Alg:new_queue} with a subroutine that is logically equivalent, which is however able to deal with multiple machines (see Algorithm~\ref{Alg:new_queue_PM}).
That is, when a machine $m$ is finished executing a task, in the case that it was a testing task for job $i$, the algorithm will insert task $p_i$ with weight $t_i + p_i$ into the queue.
Otherwise, no item is inserted into the queue.
Next, the algorithm will extract the smallest-weight task in the queue, and execute it on machine $m$.
Note that if this is task $p_i$, then $t_i$ and $p_i$ are executed consecutively on the same machine.

\begin{algorithm}[t]
\begin{algorithmic}
{ 
\renewcommand{\algorithmicfor}{\textbf{when}}
    \Procedure{Multiple Machines Queue-Execution}{$Q$}
        \While{any machine is not idle}
            \For{machine $m$ becomes idle}
                \State Let $y$ be the previous task executed on machine $m$
                \If{$y$ was a testing task for job $i$}
                    \State Insert the execution task of job $i$ with weight $t_i + p_i$ into $Q$
                \EndIf
                \State {$x \gets$ Extract the smallest-weight} task in $Q$
                \State Execute task $x$ on machine $m$
            \EndFor
        \EndWhile
    \EndProcedure

\caption{Procedure \textbf{Multiple Machines Queue-Execution} $(Q)$}
\label{Alg:new_queue_PM}
}
\end{algorithmic}
\end{algorithm}

The challenge of applying the priority-queue based algorithms is the precedence constraints between the testings and actual executions. 

\bigskip

\runtitle{An unlucky case in the multiple machines setting.} Consider only two available jobs~$k$ and~$j$, where the algorithm tests~$k$ and executes~$j$ untested, and $\beta\cdot t_k \leq t_k + p_k \leq u_j$.
In the scenario where two machines $m_1$ and $m_2$ are available at time~$\tau_1$ and $\tau_2 \geq \tau_1$, respectively, $t_k$ will be assigned to $m_1$, and $p_k$ will be available at time $\tau_1+t_k$. 
If $\tau_2 < \tau_1+t_k$, a no-wait priority-queue based algorithm will assign~$u_j$ to~$m_2$, which violates the real priority, where~$p_k$ should be assigned before~$u_j$ since $t_k+p_k < u_j$.

\bigskip

\runtitle{Priority violation.} A \emph{priority violation} occurs when task $o_x$ would have been assigned before~$o_y$ in the single-machine case, but $o_x$ starts later than~$o_y$ because~$o_x$ is available only after the start time of~$o_y$ in the multiple-machines case.
Since both testing tasks and untested execution tasks are available from the beginning of execution, only a tested execution task can cause a priority violation, i.e., $o_x = p_x$.
Formally, for any task~$o_i$, denote the time it becomes available to execute by~$a(o_i)$, the time the execution starts by~$s(o_i)$, and its weights in the priority queue by $w(o_i)$. 
A priority violation between $p_x$ and $o_y$ occurs when 
\[w(p_x) \leq w(o_y), \hspace{1cm} w(t_x) \leq w(o_y), \hspace{0.5cm}\text{ and }\hspace{0.5cm} a(p_x) > s(o_y).\]

In order to see why the condition $w(t_x) \leq w(o_y)$ is required for a priority violation to occur between $p_x$ and $o_y$, let us first consider $o_y \in \{t_y, u_y \}$.
Both tasks $o_y$ and $t_x$ are available from the start, hence if $w(t_x) > w(o_y)$, both in the single machine case as in the multiple machine case, $o_y$ is executed before~$t_x$.
Now consider $o_y = p_y$ where $w(t_x) > w(p_y)$.
If $w(t_x) < w(t_y)$, the algorithm schedules~$t_x$ before scheduling~$t_y$.
Consequently, $a(p_x) < a(p_y)$, implying that~$p_y$ cannot start before~$p_x$, and thus cannot violate its priority.
On the other hand, if $w(t_x) \geq w(t_y)$, then $t_y$ is scheduled before~$t_x$, and since $w(p_y) < w(t_x)$, task~$p_y$ has priority over $p_x$ (the single machine order would be $t_y, p_y, t_x, p_x$), resulting in no priority violation between~$p_x$ and~$p_y$.

To deal with the priority violation issue, the authors of~\cite{DBLP:journals/algorithmica/GongCH23} defined \emph{last segments} of jobs:

\begin{definition}(Recasting Definition $4$ in~\cite{DBLP:journals/algorithmica/GongCH23})
    Given a list scheduling algorithm, the \emph{last segment} of a job $j$ is 
    \begin{itemize}
        \item the task $u_j$ if $j$ is executed untested, 
        \item the tasks $t_j$ and $p_j$ if $j$ is tested, and the tasks $t_j$ and $p_j$ are on the same machine without another job in between, or
        \item the task $p_j$ otherwise.
    \end{itemize}
\end{definition}

\runtitle{Contribution of job $k$ to job $j$'s completion time.} 
By Lemma~\ref{Lem:Gong}, the authors redefined the \emph{contribution} of job $k$ to the completion time of job $j$, $c(k,j)$, by the time the algorithm spent on $k$ before the starting time of the last segment of $j$.
By the greedy nature, the completion time of job $j$ \cite{DBLP:journals/algorithmica/GongCH23}, 
\begin{equation} \label{eq:cj_parallel}
c_j\leq \frac{1}{m} \sum_{k\neq j} c(k,j) + p^A_j 
\end{equation}

\begin{lemma}[Recasting Lemmas~3 and~4 in \cite{DBLP:journals/algorithmica/GongCH23}]
    \label{Lem:Gong}
    For any priority-queue based algorithm, if task $p_x$ precedes another task $o_y \in \{ t_y, p_y, u_y \}$ in the queue, (that is, if $w(p_x) \leq w(o_y)$) but $a(p_x) > s(o_y)$, then the contribution of job $y$ to the completion time of job $x$ and vice versa is at most

    \[ c(y,x) \leq \begin{cases}
        t_y & \text{if } o_y = p_y\\
        0 & \text{if } o_y = t_y \text{ or } u_y
    \end{cases}\text{ and } c(x,y) \leq \begin{cases}
        t_x + p_x & \text{if } o_y = t_y\\
        t_x & \text{if } o_y = p_y \text{ or } u_y
    \end{cases}\]
\end{lemma}

\begin{proof}
    First, assume that the following claim is true:
    \begin{claim}
        Let $p_x$ and $o_y$ be a pair of tasks whose priority is violated (that is, $w(p_x) \leq w(o_y)$, $w(t_x) \leq w(o_y)$ and $a(p_x) > s(o_y)$).
        Then, $s(t_x) \leq s(o_y)$, and the last segment of job~$x$ is $t_x$ and $p_x$.
    \end{claim}
    
    Now, consider a priority-violated pair of tasks $p_x$ and $o_y$. 
    By the claim, the last segment of job $x$ is $t_x$ and $p_x$, and $s(t_x) \leq s(o_y)$. 
    Thus, the task $o_y$ cannot contribute to job $x$. 
    Therefore, if $o_y$ is $t_y$ or $u_y$, $c(y, x) = 0$.
    On the other hand, if $o_y$ is $p_y$, it might be the case that $s(t_y) < s(t_x)$. 
    Hence, $c(y, x) \leq t_y$ in this case. 

    Next, we prove the bound on the contribution of job $x$ to the completion time of job $y$.
    Since $s(p_x) > s(o_y)$, task $p_x$ does not contribute to the completion time of job $y$ if $o_y$ is (part of) the last segment of job $y$.        
    Thus, if $o_y$ is either $p_y$ or $u_y$, then $c(x,y) \leq t_x$.
    However, when $o_y$ is $t_y$ and the last segment of job $y$ is $p_y$ solely, $p_x$ might also contribute to the completion time of \fhl{$y$} if $s(p_x) \leq s(p_y)$.
    Hence, $c(x,y) \leq t_x + p_x$ in this case.
    
    \smallskip
    
    The last thing to do is to show the claim.    
    Assume aiming towards a contradiction that there is the first pair of priority-violated jobs $x$ and $y$ (that is, $w(p_x) \leq w(o_y)$, $w(t_x) \leq w(o_y)$ and $s(p_x) \geq a(p_x) > s(o_y)$) where $s(t_x) > s(o_y)$ or $p_x$ alone is the last segment of job $x$.
    Since any testing task $t_x$ is available since the start of the execution, $a(t_x) \leq a(o_y)$.
    Then, by $w(t_x) \leq w(o_y)$, also $s(t_x) \leq s(o_y)$, and task $p_x$ must be the last segment of job $x$ by our assumption.

    Since the priority of tasks $p_x$ and $o_y$ was violated, $a(p_x) > s(o_y)$.
    If the moment task $p_x$ was added to the queue, there existed a task $o_z$ in the queue such that $w(o_z) < w(p_x)$, task $o_y$ would violate the priority of task $o_z$, and therefore $p_x$ and $o_y$ are not the first pair of priority-violated tasks.
    Then, it follows that when $p_x$ was added to the queue it was the minimum weight element, and task $p_x$ was scheduled the moment it was available (i.e., $s(p_x) = a(p_x) = s(t_x) + t_x$) on the same machine where $t_x$ is finished.
    Thus, since $t_x$ and $p_x$ are executed consecutively on the same machine, $t_x$ and $p_x$ form the last segment of job $x$, which leads to a contradiction and proves the claim.
    \qedhere
\end{proof}

\subsection{Applying the amortized analysis on multiple machines} 
Our analysis can be further combined with the multiprocessor environment framework by Gong~\cite{DBLP:journals/algorithmica/GongCH23}:

\begin{lemma}
    \label{lem:multNewN}
    Given two jobs $k \leq_o j$, if $\paralgo$ does not test job $k$,
    \[ c(k,j) + c(j,k) \leq (1 + \frac{1}{\beta})u_k \]
\end{lemma}

\begin{proof}
    The proof is the same as the proof of Lemma~\ref{lem:newN} except for the case where $j$ is tested, and tasks $p_j$ and $u_k$ form a priority-violated pair.
    
    This priority-violation happens when $\beta t_j \leq u_k$ and $t_j + p_j \leq u_k$. 
    In this case, by Lemma~\ref{Lem:Gong} (where $x = j$ and $y = k$), $c(k, j) + c(j, k) \leq 0 + t_j \leq \frac{u_k}{\beta}$.
    \qedhere
\end{proof}

\begin{lemma}
    \label{lem:multNewT}
    Given two jobs $k \leq_o j$, if $\paralgo$ tests job $k$,
    \[ c(k,j) + c(j,k) \leq \max \{2t_k + p_k, \beta t_k, (1 + \frac{1}{\beta})(t_k + p_k) \} \]
\end{lemma}

\begin{proof}
    The proof is the same as the proof of Lemma~\ref{lem:newT} except the cases where priority-violation happens.
    More specifically, $p_k$ can form a priority-violation pair with $u_j$, $t_j$, or $p_j$.
    \begin{itemize}
        \item Priority-violation pair $p_k$ and $u_j$:
            By Lemma~\ref{Lem:Gong} with $x = k$ and $y = j$, $c(k, j) + c(j, k) \leq t_x + 0 = t_k$.
        
        \item Priority-violation pair $p_k$ and $t_j$:
            By Lemma~\ref{Lem:Gong} with $x = k$ and $y = j$, $c(k, j) + c(j, k) \leq (t_x + p_x) + 0 = t_k + p_k$.

        \item Priority-violation pair $p_k$ and $p_j$:
            This happens when $t_j \leq t_k$, $\beta t_k \leq t_j + p_j$, and $t_k + p_k \leq t_j + p_j$.
            By Lemma~\ref{Lem:Gong} with $x = k$ and $y = j$, $c(k, j) + c(j, k) \leq t_x + t_y = t_k + t_j \leq 2t_k$.
    \end{itemize}

    Similarly, if $j$ is tested by the algorithm, $p_j$ can form priority-violated pairs with $t_k$ or $p_k$ (recall that this lemma focuses on the case where $k$ is tested).
    \begin{itemize}
        \item Priority-violation pair $p_j$ and $t_k$: 
            This happens when $t_j \leq t_k$ and $t_j + p_j \leq \beta t_k$.
            By Lemma~\ref{Lem:Gong} with $x = j$ and $y = k$, $c(k, j) + c(j, k) \leq 0 + (t_x + p_x) = t_j + p_j \leq \beta t_k$.
        
        \item Priority-violation pair $p_j$ and $p_k$:
            This happens when $t_k \leq t_j$, $\beta t_j \leq t_k + p_k$, and $t_j + p_j \leq t_k + p_k$.
            By Lemma~\ref{Lem:Gong} with $x = j$ and $y = k$, $c(k, j) + c(j, k) \leq t_y + t_x = t_k + t_j \leq \frac{1}{\beta}(t_k + p_k) + t_k$.
            \qedhere
    \end{itemize} 
\end{proof}

\begin{theorem}\label{Thm:multi-p_variable}
    Let $R = \max\{\alpha(1+\frac{1}{\beta}), 1+\frac{1}{\alpha}+\frac{1}{\beta} + \frac{1}{\alpha\beta}, \beta, 2 , 1+\frac{2}{\alpha} \}$ and $r = \max\{\alpha, 1+\frac{1}{\alpha}\}$. 
    When $R \leq 2r$, there is a $R \cdot (\frac{1}{2}+\frac{1}{2m})+ r \cdot (1-\frac{1}{m})$-competitive algorithm for the \textsc{SEU} problem with objective minimizing the total completion time on $m$ multiple machines.
\end{theorem}

\begin{proof}
Similar to the proof of Theorem~\ref{Thm:improved}, using the bounds of Lemma~\ref{lem:multNewN} and Lemma~\ref{lem:multNewT},
    let $R$ be $\frac{c(k,j) + c(j,k)}{\opt{p_k}} \leq \max\{\alpha(1+\frac{1}{\beta}), 1+\frac{1}{\alpha}+\frac{1}{\beta} + \frac{1}{\alpha\beta}, \beta, 2 , 1+\frac{2}{\alpha}\}$.
    Furthermore, let $r$ be $\frac{c(k,k)}{\opt{p_k}} \leq \max\{\alpha, 1+\frac{1}{\alpha}\}$.
    
    \begin{align*}
        \sum_j C^A_j 
        & \stackrel{Eq\eqref{eq:cj_parallel}}{\leq} \sum_j \left( \frac{1}{m} \sum_{k\neq j} c(k,j) + p^A_j \right) \\
        & \stackrel{p^A_k = c(k,k)}{=} \sum_j \left( \frac{1}{m} \sum_k c(k,j) + (1-\frac{1}{m})p^A_j \right) \\
        & \leq \frac{1}{m} \sum_j\sum_{k \leq_o j} (c(k,j) + c(j,k)) + \sum_j(1-\frac{1}{m}) p^A_j\\
        & \stackrel{Lemma~\ref{Thm:improved}}{\leq} \frac{1}{m} \cdot R \cdot \sum_j \sum_{k \leq_o j} p^*_k + r \cdot (1-\frac{1}{m})\sum_j p^*_j\\ 
        & = R\cdot \left( \frac{1}{m}\sum_j\sum_{k \leq_o j} p^*_k + (\frac{1}{2} - \frac{1}{2m})\sum_j p^*_j \right) \\ &\hspace{20pt}+ \left( r\cdot(1-\frac{1}{m})-R\cdot(\frac{1}{2}- \frac{1}{2m}) \right) \sum_j p^*_j\\
        & \stackrel{R\leq 2r}{\stackrel{Eq\eqref{eq:opt_LB_parallel}}{\leq}} R\cdot \cost(\optimal) + \left( r\cdot(1-\frac{1}{m})-R\cdot(\frac{1}{2}-\frac{1}{2m}) \right) \cdot \cost(\optimal) \\
        & = \left( R\cdot(\frac{1}{2} + \frac{1}{2m}) + r\cdot(1-\frac{1}{m}) \right) \cdot \cost(\optimal).
    \end{align*}
    \qedhere
\end{proof}

\begin{corollary}
\label{Cor:multi-p_variable}
    By choosing $\alpha = \phi$ and $\beta = \frac{\phi + \sqrt{5\phi+1}}{2} \leq 2.31652$, the competitive ratio of $\paralgo$ on $m$ multiple machines is $\frac{\phi + \sqrt{5\phi+1}}{2} \cdot(\frac{1}{2}+\frac{1}{2m})+\phi\cdot(1-\frac{1}{m})$.
    When $m$ tends to infinity, the competitive ratio is $2.77630$.
\end{corollary}

\subsection{Uniform testing time}
Next, we further investigate the performance of $\paralgo$ when all jobs have uniform testing times.
Without loss of generality, we assume that for each job $j$, the testing cost equals $t_j = 1$.
Note that when $\alpha < \beta$, the schedule has a special form:
all untested jobs are executed first, before any testing begins.
Once all untested jobs are executed, the algorithm tests the remaining jobs.
Tested jobs $j$ with processing time $p_j \leq \beta - 1$ are executed immediately, while those with longer processing times are postponed until the end.
Thus, we adjust $\paralgo$ a bit by removing the parameter $\beta$. 
The algorithm first executes all jobs that do not require testing, i.e., the jobs $j$ with $u_j \leq \alpha$.
Next, it tests all the remaining jobs.
Only after completing these tests does it execute the tested jobs, in order of non-decreasing actual processing times.

\begin{theorem}\label{Thm:multi-p_uniform}
    Let $R = \max \{ 2, \alpha, 1 + \frac{2}{\alpha} \}$ and $r = \max \{ \alpha, 1 + \frac{1}{\alpha} \}$. When $R\leq 2 \cdot r$,
    and
    the testing times of all jobs are uniform, there is a $\max\{2,\alpha, 1+\frac{2}{\alpha}\}\cdot(\frac{1}{2}+\frac{1}{2m})+\max\{\alpha, 1+\frac{1}{\alpha}\}\cdot(1-\frac{1}{m})$-competitive algorithm for the \textsc{SEU} problem with objective minimizing the total completion time on $m$ multiple machines.
\end{theorem}

\begin{proof}
    Consider jobs $k$ and $j$ where $k \leq_o j$. 
    If at least one of them is not tested, i.e., $u_k \leq \alpha$ or $u_j \leq \alpha$, the contribution $c(k,j) + c(j,k) \leq \alpha$, since all untested jobs are executed before any testing begins.
    It follows that $c(k,j)+c(j,k) \leq \alpha \leq u_k$ if job $k$ is tested.
    On the other hand, if job $k$ is untested, job $j$ proceeds $k$ only when $u_j \leq u_k \leq \alpha$. 
    Therefore, in this case, $c(k,j) + c(j,k)$ is also upper-bounded by $u_k$.
      
    Now, we consider the case where both $u_k \geq \alpha$ and $u_j \geq \alpha$.
    Recall that the algorithm first executes all jobs $j$ with $u_j < 1$, then tests all remaining jobs, and finally executes the tested jobs $j$ in non-decreasing order of $p_j$'s. 
    Thus, $p_k$ and $p_j$ are the only possible pairs of a priority-violated pair, and there is no priority-violated pair where a testing task is involved.
    By Lemma~\ref{Lem:Gong}, for either the case where $x = k$ and $y = j$ or the other case, $x = j$ and $y = k$, $c(k, j) + c(j, k) \leq t_k + t_j = 2$.
    Therefore, similar to the proof of Lemmas~\ref{lem:newN} and~\ref{lem:newT}, $c(k,j) + c(j,k) \leq \max \{u_k, 2 + p_k, 2\}$.

    Then, $\frac{c(k,j) + c(j,k)}{\opt{p_k}} \leq \max\{\alpha, 2, 1+\frac{2}{\alpha}\}$ and $\frac{c(k,k)}{\opt{p_k}} \leq \max\{\alpha, 1+ \frac{1}{\alpha}\}$. 
    By the same techniques in the proof of Theorem~\ref{Thm:amortized} and Theorem~\ref{Thm:multi-p_uniform}, the cost of the algorithm is at most 
    \[\left(R\cdot(\frac{1}{2}+\frac{1}{2m}) + r\cdot(1-\frac{1}{m})\right) \cdot \cost(\optimal).\]
    \qedhere
\end{proof}

\begin{corollary}
    By choosing $\alpha = \phi$ and $\beta = \frac{\phi+\sqrt{5\phi+1}}{4} \leq 2.31652$, the competitive ratio of $\paralgo$ on $m$ multiple machines is $\sqrt{5} \cdot (\frac{1}{2}+\frac{1}{2m}) + \phi \cdot (1-\frac{1}{m})$ for instances with uniform testing times.
    When $m$ tends to infinity, the competitive ratio is $2.73606$.
\end{corollary}

Interestingly, the later results on uniform testing time suggest that when $m = 1$ and $2$, the parameter $\alpha$ should be chosen as $2$ and $\sqrt{3}$, respectively. And it shows that the competitive ratio of the algorithm is $2$ and $2.48206$, respectively.
Note that in Theorems~\ref{Thm:multi-p_variable} and~\ref{Thm:multi-p_uniform}, the ratios when $m= 1$ match the current best results on a single machine.

\subsection{Randomized algorithm}
The framework for multiple machines proposed by Gong et al.~\cite{DBLP:journals/algorithmica/GongCH23} can also be applied to randomized algorithms.

\runtitle{The algorithm.}
For multiple machines, we base our randomized algorithm, called $\randparalgo$, on $\randalgo$ (Algorithm~\ref{Alg:randomized}).
We replace the subroutine described in Algorithm~\ref{Alg:new_queue} by a subroutine that is logically equivalent, however is able to deal with multiple machines (Algorithm~\ref{Alg:new_queue_PM}).

\begin{lemma}
    \label{Lem:expContMult}
    If $\prob_k$ is the probability that $\randparalgo$ tests job $k$, then
    \[\expect[c(k,j) + c(j,k)] \leq (1 + \frac{1}{\beta}) \cdot u_k \cdot (1 - \prob_k) + \max\{ 2t_k + p_k, \beta t_k, t_k + (1 + \frac{1}{\beta})\cdot p_k\}\cdot \prob_k\]
\end{lemma}

\begin{proof}
    By Lemma~\ref{lem:multNewN}, $\expect [c(k,j) + c(j,k) \mid k \text{ is not tested}] \leq (1 + \frac{1}{\beta}) u_k$.
    And by Lemma~\ref{lem:multNewT}, $\expect [c(k,j) + c(j,k) \mid k \text{ is tested}] \leq  \max\{ 2t_k + p_k, \beta t_k, t_k + (1 + \frac{1}{\beta})p_k\}$.
    Therefore, as $\prob_k$ denotes the probability that $k$ is tested, it follows that $\expect[c(k,j) + c(j,k)] \leq (1 + \frac{1}{\beta}) u_k (1 - \prob_k) + \max\{ 2t_k + p_k, \beta t_k, t_k + (1 + \frac{1}{\beta})p_k\}\prob_k$.
\end{proof}

\begin{theorem}
    \label{Thm:randomized_parallel}
    Let,
    \[\factorOne = \max\limits_k \frac{(1 + \frac{1}{\beta}) u_k (1 - \prob_k) + \max\{ 2t_k + p_k, \beta t_k, t_k + (1 + \frac{1}{\beta})p_k\}\prob_k}{\opt{p_k}} \text{, and}\]
    \[\factorTwo = \max\limits_j \frac{u_j(1 - \prob_j) +  (t_j + p_j)\prob_j}{\opt{p_j}}\text{.}\]
    When $R \leq 2r$, the expected competitive ratio of $\randparalgo$ is at most
    \[\factorOne \cdot (\frac{1}{2} + \frac{1}{2m} ) + \factorTwo \cdot (1 - \frac{1}{m})
    \]
\end{theorem}

\begin{proof}
The proof is similar to the proof of Theorem~\ref{Thm:randomized}.
By Lemma~\ref{Lem:expContMult}, $\expect[c(k,j) + c(j,k)] \leq \ (1+\frac{1}{\beta})u_k(1-\prob_k) + \max\{2t_k + p_k, \beta t_k, (1+\frac{1}{\beta})(t_k+p_k)\}\prob_k$.

Recall that, for multiple machines, the completion time of job $j$ can be bounded by $C_j \leq \frac{1}{m} \sum\limits_{k\neq j} c(k,j) + \alg{p_j}$.
Thus, it follows by linearity of expectation that $\expect[C_j] \leq \frac{1}{m} \sum\limits_{k\neq j} \expect[c(k,j)] + \expect[\alg{p_j}]$ and, $\expect[\sum\limits_{j=1}^n C_j] = \sum\limits_{j=1}^n \expect[C_j]$.

Then,
{
\allowdisplaybreaks 
\begin{align*}
    &\expect[\sum\limits_{j=1}^n C_j] \leq \sum\limits_{j=1}^n \left[ \frac{1}{m} \sum\limits_{k\neq j} \expect[c(k,j)] + \expect[\alg{p_j}] \right] \\
    &= \frac{1}{m} \sum\limits_{j=1}^n \sum\limits_{k \leq_o j} \Bigl[ \expect[c(k,j)] + \expect[c(j,k)] \Bigr] + (1 - \frac{1}{m}) \sum\limits_{j=1}^n \expect[\alg{p_j}]  \\
    &\leq \frac{1}{m} \sum\limits_{j=1}^n \sum\limits_{k \leq_o j} \Bigl[  (1 + \frac{1}{\beta}) u_k (1 - \prob_k) + \max\{ 2t_k + p_k, \beta t_k, t_k + (1 + \frac{1}{\beta})p_k\}\prob_k \Bigr] \\&\hspace{2cm}+ (1 - \frac{1}{m}) \sum\limits_{j=1}^n \left[ u_j(1 - \prob_j) +  (t_j + p_j)\prob_j \right] \\
    &\leq \left( \max\limits_k \frac{(1 + \frac{1}{\beta}) u_k (1 - \prob_k) + \max\{ 2t_k + p_k, \beta t_k, t_k + (1 + \frac{1}{\beta})p_k\}\prob_k}{\opt{p_k}} \right) \frac{1}{m} \sum\limits_{j=1}^n \sum\limits_{k \leq_o j} \opt{p_k} \\&\hspace{2cm} + \left( \max\limits_j \frac{u_j(1 - \prob_j) +  (t_j + p_j)\prob_j}{\opt{p_j}} \right) (1 - \frac{1}{m}) \sum\limits_{j=1}^n \opt{p_j} \\
\end{align*}
}

\[\text{Let, } \factorOne = \max\limits_k \frac{(1 + \frac{1}{\beta}) u_k (1 - \prob_k) + \max\{ 2t_k + p_k, \beta t_k, t_k + (1 + \frac{1}{\beta})p_k\}\prob_k}{\opt{p_k}} \text{, and}\]
\[ \factorTwo = \max\limits_j \frac{u_j(1 - \prob_j) +  (t_j + p_j)\prob_j}{\opt{p_j}} \]
Then, 

\begin{align*}
    \expect[\sum\limits_{j=1}^n C_j] &\leq \factorOne \cdot \frac{1}{m} \sum\limits_{j=1}^n \sum\limits_{k \leq_o j} \opt{p_k} + \factorTwo \cdot (1 - \frac{1}{m}) \sum\limits_{j=1}^n \opt{p_j} \\
    &= \factorOne \cdot \left( \frac{1}{m} \sum\limits_{j=1}^n \sum\limits_{k \leq_o j} \opt{p_k} + (\frac{1}{2} - \frac{1}{2m} ) \sum\limits_{j=1}^n \opt{p_j} \right) \\
    &\hspace{20pt}- \factorOne \cdot (\frac{1}{2} - \frac{1}{2m} ) \sum\limits_{j=1}^n \opt{p_j} + \factorTwo \cdot (1 - \frac{1}{m}) \sum\limits_{j=1}^n \opt{p_j} \\
    &= \factorOne \cdot \left( \frac{1}{m} \sum\limits_{j=1}^n \sum\limits_{k \leq_o j} \opt{p_k} + (\frac{1}{2} - \frac{1}{2m} ) \sum\limits_{j=1}^n \opt{p_j} \right) \\
    &\hspace{20pt}+ \left( \factorTwo \cdot (1 - \frac{1}{m}) - \factorOne \cdot (\frac{1}{2} - \frac{1}{2m}) \right) \cdot \sum\limits_{j=1}^n \opt{p_j} \\
    &\stackrel{R\leq 2r}{\stackrel{Eq\eqref{eq:opt_LB_parallel}}{\leq}} \factorOne \cdot \cost(\optimal) + \left( \factorTwo \cdot (1 - \frac{1}{m}) - \factorOne \cdot (\frac{1}{2} - \frac{1}{2m}) \right) \cdot \cost(\optimal) \\
    &= \left(\factorOne \cdot (\frac{1}{2} + \frac{1}{2m} ) + \factorTwo \cdot (1 - \frac{1}{m}) \right) \cdot \cost(\optimal) \\
\end{align*}
\qedhere
\end{proof}

\begin{corollary}
    By choosing $\beta = 2$ and  $\prob_k = 0$ if $r_k < 1$, $\prob_k = 1$ if $r_k > 3$ and 
    \[ \prob_k = \frac{(\beta + 1)(r_k - 1)}{\beta(\max \{ \frac{2}{r_k} + 1, \frac{\beta}{r_k}, (1 + \frac{1}{\beta})(1 + \frac{1}{r_k})\} - \max \{ 2, \beta, 1 + \frac{1}{\beta} \} + r_k - 1) + r_k - 1} \]
    if $r_k \in [1,3]$, $\randparalgo$ has expected competitive ratio  
    \[ 2.152271 \cdot (\frac{1}{2} + \frac{1}{2m} ) + 1.434847 \cdot (1 - \frac{1}{m})\] which is $2.152271$ when $m=1$ and $2.510983$ when $m$ tends to infinity.
\end{corollary}

\begin{proof}
From Lemma~\ref{Lem:bestBeta} it follows that $\factorOne$ is minimized for $\beta = 2$.
For this case the probability can be simplified to $\prob_k = \frac{3r_j^2 - 3r_j}{3r_j^2-4r_j+3}$, and $\factorOne \leq 2.152271$.
It remains to bound $\factorTwo$.
\begin{itemize}
    \item First, assume $\opt{p_j} = u_j$, then,
\begin{align*}
    \frac{u_j(1 - \prob_j) +  (t_j + p_j)\prob_j}{\opt{p_j}} &\leq (1 - \prob_j) +  (\frac{t_j}{u_j} + 1)\prob_j \\
    &= (1 - \prob_j) +  (\frac{1}{r_j} + 1)\prob_j \\
    &= 1 + \frac{\prob_j}{r_j}
\end{align*}

Using the formula for $\prob_j = \frac{3r_j^2 - 3r_j}{3r_j^2-4r_j+3}$, it follows that 
$1 + \frac{\prob_j}{r_j} = 1 + \frac{3r_j - 3}{3r_j^2 - 4r_j + 3}$.
This function has maximum $\frac{7 + 3\sqrt{6}}{10} \leq 1.434847$ when $r_j = 1 + \sqrt{\frac{2}{3}}$.

    \item Next, assume $\opt{p_j} = t_j + p_j$, then,
\begin{align*} 
    \frac{u_j(1 - \prob_j) +  (t_j + p_j)\prob_j}{\opt{p_j}} &\leq r_j(1 - \prob_j) + \prob_j \\
    &= r_j - r_j\prob_j + \prob_j
\end{align*}

Using the formula for $\prob_j = \frac{3r_j^2 - 3r_j}{3r_j^2-4r_j+3}$, it follows that 
$r_j - r_j\prob_j + \prob_j = \frac{2r_j^2}{3r_j^2 - 4r_j + 3}$.
This function has maximum $\frac{6}{5}=1.2$ when $r_j = \frac{3}{2}$.
\end{itemize}
It follows that $\frac{u_j(1 - \prob_j) +  (t_j + p_j)\prob_j}{\opt{p_j}}$ is at most $1.434847$.
Therefore, if we let $r = 1.434847$ and $R = 2.152271$, which is less than $2r$. By Theorem~\ref{Thm:randomized_parallel}, the expected cost of $\randparalgo$ is at most
\[ \left( 2.152271 \cdot (\frac{1}{2} + \frac{1}{2m} ) + 1.434847 \cdot (1 - \frac{1}{m}) \right) \cdot \cost(\optimal) \]
Which is $2.152271$ when $m=1$ and $2.510983$ when $m$ tends to infinity.
\qedhere
\end{proof}

%% file: Conclusion.tex
In this work, we study a scheduling problem with explorable uncertainty.
We enhance the analysis framework proposed in the work~\cite{DBLP:conf/waoa/AlbersE20} by introducing amortized perspectives. 
Using the enhanced analysis framework, we are able to balance the penalty incurred by different wrong decisions of the online algorithm.
In the end, we improve the competitive ratio on a single machine significantly from $4$ to $2.316513$ (deterministic) and from $3.3794$ to $2.152271$ (randomized).
An immediate open problem is if one can further improve the competitive ratio by a deeper level of amortization. 

Additionally, we show that preemption does not improve the competitive ratio in the current problem setting, where all jobs are available at first. 
It may not be true in the fully online setting, where jobs can arrive at any time. Thus, another open problem is to study the problem in the fully online model.

Furthermore, we proposed an algorithm for multiple machines
with competitive ratio $2.77629-(0.45977/m)$ (deterministic) and $2.51098-(0.3587/m)$ (randomized).
We note that recently another group in parallel has also extended our work to multiple machines~\cite{DBLP:conf/faw/BuldS25}.



%% file: Appendix.tex
\runtitle{Our contribution on makespan.}
In addition to the total completion time, we also study the objective of minimizing the makespan, where the highest load of machines is minimized.
Albers and Eckl \cite{DBLP:conf/wads/AlbersE21} examined this problem in the non-preemptive setting, in which once a job is tested, its actual processing time must follow its testing time immediately on the same machine where the job is tested. 
By making upper limits extremely large and forcing algorithms to test all jobs, the authors derived a competitive ratio lower bound of $2 - \frac{1}{m}$.
However, the adversary overlooked the impact of uncertain processing time on the competitive ratio.
We carefully select an appropriate upper limit that puts online algorithms in a predicament and thus improves the lower bound.

\begin{theorem}
    For \textsc{SEU} problem under the non-preemptive setting and aiming at minimizing the makespan,
    the competitive ratio is at least $2 - \frac{1}{2m}$
    even for uniform testing time and uniform upper limit.
\end{theorem}
\begin{proof}
    We present an adversary that makes any algorithm~$A$ be at least $(2 - 1 / (2 m))$-competitive.
    
    The adversary generates an instance that depends on the behavior of~$A$.
    It first releases $2 m (m - 1) + 1$ jobs, each with testing time~$t_j = 1$ and upper limit~$u_j = 2 m$.
    Due to the large amount of jobs, $A$ must assign at least one of them to time $2 (m - 1) + 1$ or later.
    This is because assigning the jobs, each occupies at least time $1$, on $m$ machines requires $\lceil (2 m (m - 1) + 1) / m \rceil = 2 (m - 1) + 1$ time.
    Afterwards, the adversary sets all $p_j = 0$ for all jobs $j$ except job $j'$, which is any job picked by the adversary from the jobs assigned to time $2 (m - 1) + 1$ or later.
    Depending on $A$'s behavior on $j'$, there are two cases:

    Case 1: $A$ tests $j'$.
    In this case, the adversary sets $p_{j'} = 2 m$, i.e., it makes $j'$ have the processing time equal to the upper limit.
    Since jobs must be scheduled non-preemptively, the completion time of $j'$ is at least $2 (m - 1) + 1 + p_{j'} = 4 m - 1$.
    In contrast, $\optimal$ does not test $j'$ and schedules it solely on one of the machines.
    The other $2 m (m - 1)$ jobs are all tested and assigned evenly on the other machines.
    The $\optimal$ schedule has load $2 m$ for all machines.
    Thus, the competitive ratio of $A$ is at least $(4 m - 1) / (2 m) = 2 - 1 / (2 m)$.

    Case 2: $A$ does not test $j'$.
    The completion time of $j'$ is at least $2 (m - 1) + u_{j'} = 4 m - 2$.
    To benefit $\optimal$, the adversary sets $p_{j'} = 0$.
    $\optimal$ tests all the jobs and assigns them evenly.
    The makespan of $\optimal$ is $\lceil (2 m (m - 1) + 1) / m \rceil = 2 m - 1$.
    Thus, the competitive ratio of $A$ is at least $(4 m - 2) / (2 m - 1) = 2$.

    Finally, the lower bound of competitive ratio is obtained by the smaller one of the two cases, which is $2 - 1 / (2 m)$.
    \qed
\end{proof}

The lower bound indicates that \textsc{SEU} problem with makespan minimization is more uncertain than its counterpart in the pure online model, where a $(2 - \frac{1}{m})$-competitive algorithm exists.